\RequirePackage{fix-cm}
\documentclass[twocolumn]{svjour3}          
\smartqed  

\def\presuper#1#2%
  {\mathop{}%
   \mathopen{\vphantom{#2}}^{#1}%
   \kern-2\scriptspace%
   #2}
 
\usepackage{graphicx}
\usepackage{bbm}
\usepackage{psfrag}
\usepackage{tikz}
\usepackage{dsfont}
\usepackage{enumerate}
\usepackage{bm}
\usepackage[square]{natbib}  
\usepackage{amsmath}
\usepackage[utf8]{inputenc}
\usepackage[english]{babel}
\usepackage{amssymb}
\usepackage{algpseudocode,algorithm,algorithmicx}
\usepackage{url}
\DeclareMathOperator*{\argmax}{arg\,max}

\newcommand*\Let[2]{\State #1 $\gets$ #2}
\algrenewcommand\algorithmicrequire{\textbf{Initialisation:}}
\algrenewcommand\algorithmicensure{\textbf{Output:}}

 \journalname{}
\begin{document}

\title{Weight-Preserving Simulated Tempering 
}


\author{Nicholas G. Tawn \and Gareth O.\ Roberts \and Je{f}frey S.\ Rosenthal 
}


\institute{Nicholas G. Tawn  \at
              Department of Statistics \\
							University of Warwick \\
							United Kingdom \\
							CV4 7AL \\
              \email{n.tawn.1@warwick.ac.uk}           
           \and
              Gareth O.\ Roberts  \at
              Department of Statistics \\
							University of Warwick \\
							United Kingdom \\
							CV4 7AL \\
							Tel.: +44(0)24 7652 4631\\
              \email{Gareth.O.Roberts@warwick.ac.uk}  
						\and
						  Je{f}frey S.\ Rosenthal \at
								Department of Statistical Sciences \\
								University of Toronto \\
								100 St. George Street, Room 6018 \\
								Toronto, Ontario \\
								Canada M5S 3G3\\
								Tel.:(416) 978-4594 \\
              \email{Jeff@math.toronto.edu} 
}

\date{Received: date / Accepted: date}

\maketitle

\begin{abstract}

Simulated tempering is popular method of allowing MCMC algorithms to
move between modes of a multimodal target density $\pi$.  One problem
with simulated tempering for multimodal targets is that the weights of
the various modes change for different inverse-temperature values,
sometimes dramatically so.  In this paper, we provide a fix to overcome
this problem, by adjusting the mode weights to be preserved (i.e.,
constant) over different inverse-temperature settings.  We then apply
simulated tempering algorithms to multimodal targets using our mode
weight correction.  We present simulations in which our
weight-preserving algorithm mixes between modes much more successfully
than traditional tempering algorithms.  We also prove a diffusion limit
for an version of our algorithm, which shows that under appropriate
assumptions, our algorithm mixes in time $O(d [\log d]^2)$.
\keywords{Simulated Tempering \and Parallel Tempering \and MCMC \and Multimodality \ and Monte Carlo}
\end{abstract}


\section{Introduction}

Consider the problem of drawing samples from a target distribution,
$\pi(x)$ on a $d$-dimensional state space $\mathcal{X}$ where $\pi(\cdot)$
is only known up to a scaling constant. A popular approach is to use
Markov chain Monte Carlo (MCMC) which uses a Markov chain that is designed
in such a way that the invariant distribution of the chain is
$\pi(\cdot)$. 

However, if $\pi(\cdot)$ exhibits multimodality, then the majority of
MCMC algorithms which use tuned localised proposal mechanisms, e.g.\
\cite{roberts1997weak} and \cite{roberts2001optimal}, fail to explore the
state space, which leads to biased samples. Two approaches to overcome this
multimodality issue are the \textit{simulated} and \textit{parallel
tempering algorithms}. These methods augment the state space with auxiliary
target distributions that enable the chain to rapidly traverse the entire
state space.

The major problem with these auxiliary targets is that in general they don't preserve regional mass, see \cite{woodard2009conditions}, \cite{woodard2009sufficient} and \cite{Bhatnagar2016}. This problem can result in the required
run-time of the simulated and parallel tempering algorithms growing
exponentially with the dimensionality of the problem.

In this paper, we provide a fix to overcome this problem, by adjusting
the mode weights to be preserved (i.e., constant) over different
inverse-temperatures.  We apply our mode weight correction to produce
new simulated and parallel tempering algorithms for multimodal target
distributions.  We show that, assuming the chain mixes at the
hottest temperature, our mode-preserving algorithm will mix well for the
original target as well.

This paper is organised as follows. Section~\ref{sec:parallel}
reviews the simulated and parallel tempering algorithms and the existing
literature for their optimal setup. Section~\ref{sec:powertem}
describes the problems with modal weight preservation that are inherent
with the traditional approaches to tempering, and introduces a prototype
solution called the HAT algorithm that is similar to the parallel
tempering algorithm but uses novel auxiliary targets.
Section~\ref{sec:simstudy} presents some simulation studies of the new
algorithms. Section~\ref{sec:Theoreticalanalysis} provides a theoretical
analysis of a diffusion limit and the resulting computational
complexity of the HAT algorithm in high dimensions.
Section~\ref{sec:conc} concludes and provides a discussion of further work.

\section{Tempering Algorithms}
\label{sec:parallel}

There is an array of methodology available to overcome the issues of multimodality in MCMC,  the majority of which use state space augmentation e.g.\ \cite{Wang1990a}, \cite{geyer1991markov}, \cite{marinari1992simulated}, \cite{Neal1996}, \cite{kou2006discussion}, \cite{2017arXiv170805239N}. Auxiliary distributions that allow a Markov chain to explore the entirety of the state space are targeted, and their mixing information is then passed on to aid inter-modal mixing in the desired target. A convenient approach for augmentation methods,
such as the popular simulated tempering (ST) and parallel tempering (PT) algorithms introduced in \cite{geyer1991markov} and \cite{marinari1992simulated},
is to use power-tempered target distributions,
for which the target distribution at inverse temperature level $\beta$
is defined as \[\pi_\beta(x)\propto \left[\pi(x)\right]^\beta\]
for $\beta \in (0,1]$.
For each algorithm one needs to choose a sequence of $n+1$ ``inverse temperatures'' such that $\Delta:=\{\beta_0,\ldots,\beta_n\}$ where $0 \leq \beta_n<\beta_{n-1}<\ldots <\beta_1<\beta_0=1$,
so that $\pi_{\beta_0}$ equals the original target density $\pi$, and
hopefully the hottest distribution $\pi_{\beta_n}(x)$ is sufficiently
flat that it can be easily sampled.

The ST algorithm augments the original state space with a single dimensional component indicating the current inverse temperature level creating a $(d+1)$ - dimensional chain, $(\beta,X)$, defined on the  extended state space $\{\beta_0,\ldots,\beta_n\} \times \mathcal{X}$ that targets 
\begin{equation}
\pi(\beta,x)\propto K(\beta)\pi(x)^{\beta}
\label{eq:simtarg}
\end{equation}
where ideally $K(\beta)=\left[\int_{x}\pi(x)^{\beta}\mbox{d}x\right]^{-1}$, resulting in a uniform marginal distribution over the temperature component of the chain. Techniques to learn  $K(\beta)$ exist in the literature, e.g.\ \cite{Wang2001a} and \cite{Atchade2004}, but these techniques can be misleading unless used with care. The ST algorithm procedure is given in Algorithm~\ref{alg:ST}.

\begin{algorithm}
  \caption{The Simulated Tempering (ST) Algorithm
    \label{alg:ST}}
  \begin{algorithmic}[1]
    \Require{A temperature schedule $\Delta$;  initialising chain value, $(\beta_{T^0},x^0)$;  a within temperature proposal mechanism, $q_\beta(x,\cdot)$; $s$, the number of algorithm iterations  and $m$, the number of within-temperature proposals.}
    \Statex
    \Function{ST}{$\Delta, x^0, \beta_0$}
      \For{$i \gets 1 \textrm{ to } s$}
			\Let{$t$}{$(i-1)+(i-1)(m+1)$}
        \Let{$w$}{Unif$\{ -1 , 1 \}$}
				\Let{$T^{'}$}{ $T^t$+w}
				\State{Compute:\begin{equation}
A=\mbox{min}\Bigg( 1,\frac{K(\beta_{T^{'}})\pi(x^t)^{\beta_{T^{'}}}}{K(\beta_{T^t})\pi(x^t)^{\beta_{T^t}}} \Bigg).
\label{eq:accep}
\end{equation}}
        \State{Sample $U \sim \mbox{Unif}(0,1)$}
				\If{$U \le A$} 
				\Let{($\beta_{T^{t+1}},x^{t+1}$)}{($\beta_{T^{'}},x^{t}$)}
				\Else \Let{($\beta_{T^{t+1}},x^{t+1}$)}{($\beta_{T^{t}},x^{t}$)}
        \EndIf
				\State{Perform $m$ updates to the $\mathcal{X}$-marginal according to $q_{\beta_{T^{t+1}}}(x,\cdot)$ to get $\{x^{t+2} ,\ldots, x^{t+m+1} \}$.}
      \EndFor
      \State \Return{$\{(\beta_{T^{0}},x^0),(\beta_{T^{1}},x^1),\ldots, (\beta_{T^{s+s(m+1)}},x^{s+s(m+1)})\}$}
    \EndFunction
  \end{algorithmic}
\end{algorithm}

The PT approach is designed to overcome the issues arising due to  the typically unknown marginal normalisation constants. The PT algorithm runs a  Markov chain on an augmented state space $\mathcal{X}^{(n+1)}$ with target distribution defined as
\begin{eqnarray}
\pi_n(x_0,x_1,\ldots,x_n) \propto \pi_{\beta_0}(x_0)\pi_{\beta_1}(x_1)\ldots\pi_{\beta_n}(x_n).\nonumber
\end{eqnarray}
The PT algorithm procedure is given in Algorithm~\ref{alg:PT}.

\begin{algorithm}
  \caption{The Parallel Tempering (PT) Algorithm
    \label{alg:PT}}
  \begin{algorithmic}[1]
    \Require{A temperature schedule $\Delta$;  initialising chain values, $X^0=\{x_0^0,x^0_1,\ldots,x^0_n\}$; a within temperature proposal mechanism, $q_\beta(x,\cdot)$; $s$, the number of algorithm iterations  and $m$, the number of within-temperature proposals.}
    \Statex
    \Function{PT}{$\Delta, X^0$}
      \For{$i \gets 1 \textrm{ to } s$}
			\Let{$t$}{$(i-1)+(i-1)(m+1)$}
        \State{Sample $k$ uniformly from $\{ 0,1,\ldots,(n-1) \}$}
				\State{Compute:\begin{equation}
A=\mbox{min}\Bigg( 1,\frac{\pi_{\beta_{k+1}}(x^t_k)\pi_{\beta_k}(x^t_{k+1})}{\pi_{\beta_k}(x^t_k)\pi_{\beta_{k+1}}(x^t_{k+1})} \Bigg).
\label{eq:parstd1}
\end{equation}}
        \State{Sample $U \sim \mbox{Unif}(0,1)$}
				\If{$U \le A$} 
				\Let{$X^{t+1}$}{$\{x_0^t,\ldots,x_{k+1}^t,x_{k}^t,\ldots,x^t_n\}$}
				\Else \Let{$X^{t+1}$}{$X^t$}
        \EndIf
				\For{$p \gets 0 \textrm{ to } n$}
				\State{$m$ updates to the $p^{\text{th}}$-marginal chain according to $q_{\beta_p}(x,\cdot)$ to get $\{x_{p}^{t+2} ,\ldots, x_{p}^{t+m+1} \}$.}
				\EndFor
      \EndFor
      \State \Return{$\{X^0,X^1,\ldots, X^{s+s(m+1)}\}$}
    \EndFunction
  \end{algorithmic}
\end{algorithm}

\subsection{Optimal Scaling for the ST and PT Algorithms}



\cite{atchade2011towards} and \cite{roberts2014minimising}
investigated the problem of selecting
optimal inverse-temperature spacings for the ST and PT algorithms.
Specifically, if a move between two
consecutive temperature levels, $\beta$ and $\beta'=\beta+\epsilon$, is
to be proposed, then what is the optimal choice of $\epsilon$?
Too large, and the move will probably be rejected; too small, and the
move will accomplish little (similar to the situation for the
Metropolis algorithm, cf.\ \cite{roberts1997weak} and
\cite{roberts2001optimal}).


For ease of analysis,
\cite{atchade2011towards} and \cite{roberts2014minimising}
restrict to $d$-dimensional target distributions of the iid form:
\begin{equation}
\pi(x) \propto \prod_{i=1}^d f(x_i).
\label{eq:restrict}
\end{equation}
They assume that the process mixes immediately (i.e., infinitely
quickly) {\it within} each temperature, to allow them to concentrate
solely on the mixing of the inverse-temperature process itself.
To achieve non-degeneracy of the limiting behaviour of the
inverse-temperature process as $d \rightarrow \infty$, the spacings are
scaled as $O(d^{-1/2})$, i.e.\ $\epsilon = \ell/d^{1/2}$ where $\ell =
\ell(\beta)$ a positive value to be chosen optimally.

Under these assumptions,
\cite{atchade2011towards} and \cite{roberts2014minimising} prove that
the inverse-temperature processes of
the ST and PT algorithms converge, when speeded up by a factor of $d$,
to a specific diffusion limit, independent of dimension, which thus mixes
in time $O(1)$, implying that the original
ST and PT algorithms mix in time $O(d)$ as $d\to\infty$.
They also prove that
the mixing times of the ST and PT algorithms are optimised
when the value of $\ell$ is chosen to maximise the quantity
\[\ell^2 \times 2 \Phi \left( -\ell
\frac{\sqrt{I(\beta)}}{2} \right)\]
where $I(\beta)=\text{Var}_{\pi^\beta}\big( \log f(x) \big)$.
Furthermore, this optimal choice of $\ell$ corresponds to an acceptance
rate of inverse-temperature moves of 0.234 (to three decimal places),
similar to the earlier Metropolis algorithm results of \cite{roberts1997weak}
and \cite{roberts2001optimal}.

From a practical perspective, setting up the temperature levels to
achieve optimality can be done via a stochastic approximation approach
(\cite{robbins1951stochastic}), similarly to
\cite{miasojedow2013adaptive} who use an adaptive MCMC framework
(see e.g.\ \cite{roberts2009examples}).

\subsection{Torpid Mixing of ST and PT Algorithms}
\label{subsec:torpid}

The above optimal scaling results suggest that the mixing time of the ST
and PT algorithms through the temperature schedule is $O(d)$, i.e.\
grows only linearly with the dimension of the problem, which is very
promising. However, this optimal, non-degenerate scaling was derived
under the assumption of immediate, infinitely fast within-temperature
mixing, which is almost certainly violated in any real application.
Indeed, this assumption appears to be overly strong once one considers
the contrasting results regarding the scalability of the ST algorithm
from \cite{woodard2009conditions} and \cite{woodard2009sufficient}.
Their approach instead relies on a detailed analysis of the spectral gap
of the ST Markov chain and how it behaves asymptotically in dimension.
They show that in cases where the different modal structures/scalings
are distinct, this can lead to mixing times that grow exponentially in
dimension, and one can only hope to attain polynomial mixing times in
special cases where the modes are all symmetric.

The fundamental issue with the ST/PT approaches are that in cases where
the modes are not symmetric, the tempered targets do not preserve the
regional/modal weights.  That motivates the current work, which is
designed to preserve the modal weights even when performing tempering
transformations, as we discuss next.

Interestingly, a lack of modal symmetry in the
multimodal target will affect essentially all the standard multimodal
focused methods: the Equi-Energy Sampler of \cite{kou2006discussion},
the Tempered Transitions of \cite{Neal1996}, and the Mode Jumping
Proposals of \cite{Tjelmeland2001}, all suffer in this setting. Hence,
the work in this paper is applicable beyond the immediate setting of the
ST/PT approaches.

\section{Weight Stabilised Tempering}
\label{sec:powertem}


In this section, we present our modifications which preserve the weights
of the different modes when performing tempering transformations.  We
first motivate our algorithm by considering mixtures of Gaussian
distributions.

Consider a $d$-dimensional bimodal Gaussian target distribution with
means, covariance matrices and weights given by $\mu_i,~\Sigma_i,~ w_i
$ for $i=1,2$ respectively. Hence the target density is given by:
\begin{eqnarray}
\pi(x)  =  w_1 \phi(x,\mu_1,\Sigma_1)+ w_2 \phi(x,\mu_2,\Sigma_2), \label{eq:bimod}
\end{eqnarray}
where $\phi(x,\mu,\Sigma)$ is the pdf of a multivariate Gaussian with
mean $\mu$ and covariance matrix $\Sigma$.
Assuming the modes are well-separated then 
the power tempered target from~\eqref{eq:simtarg}
can be approximated by a bimodal
Gaussian mixture (cf.\ \cite{woodard2009sufficient},
\cite{NTawnThesis}):
\begin{eqnarray}
\pi(x)  =  W_{(1,\beta)} \phi\left(x,\mu_1,\frac{\Sigma_1}{\beta}\right)+ W_{(2,\beta)} \phi \left(x,\mu_2,\frac{\Sigma_2}{\beta}\right), \label{eq:bimodapps}
\end{eqnarray}
where the weights are approximated as
\begin{eqnarray}
W_{(i,\beta)} &=&\frac{  w_i^{\beta} |\Sigma_i|^\frac{1-\beta}{2}  }  {w_1^{\beta} |\Sigma_1|^\frac{1-\beta}{2}+ w_2^{\beta} |\Sigma_2|^\frac{1-\beta}{2}}\nonumber \\ &\propto& w_i^{\beta} |\Sigma_i|^\frac{1-\beta}{2}.
\label{eq:keyweight}
\end{eqnarray}


A one-dimensional example of this is illustrated in Figure~\ref{Fig:badweight},
which shows plots of a bimodal Gaussian mixture distribution
as $\beta\rightarrow 0$.
Clearly the second mode, which was originally wide but very short and
hence of low weight, takes on larger and larger weight as $\beta\to 0$,
thus distorting the problem and making it very difficult for a
tempering algorithm to move from the second mode to the first when
$\beta$ is small.

%
\begin{figure}[ht]
\begin{center}
\includegraphics[keepaspectratio,width=8cm]{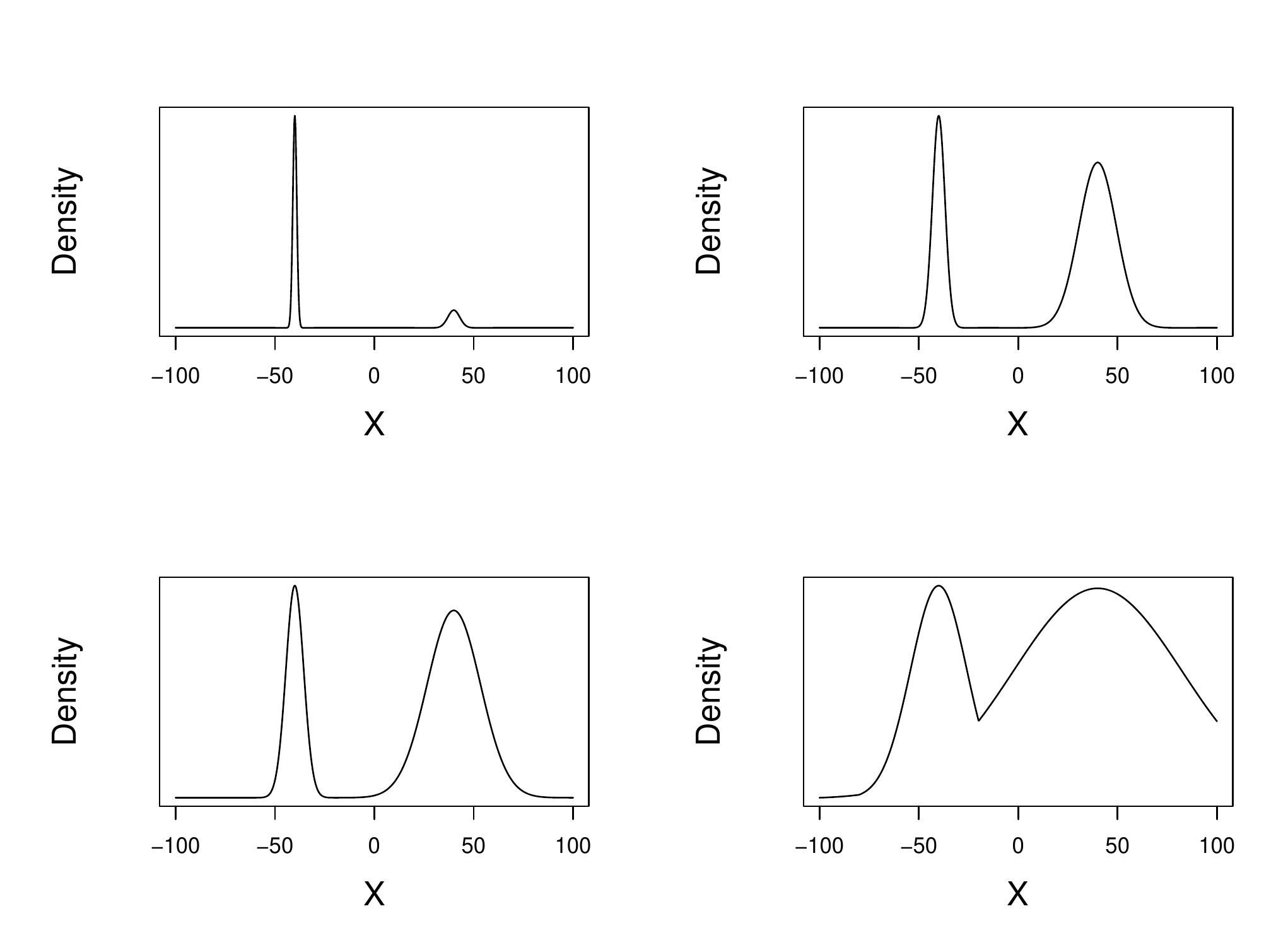}
\caption{Power-tempered target densities of a bimodal Gaussian mixture using inverse temperature levels $\beta=\{ 1,0.1,0.05,0.005 \}$ respectively. At the hot state it is evident that the mode centred on 40 begins to dominate the weight as $\beta \rightarrow \infty$.}
\label{Fig:badweight}
\end{center}
\end{figure}


Higher dimensionality makes this weight-distorting
issue exponentially worse. Consider the
situation with $w_1=w_2$ but $\Sigma_1=I_d$ and $\Sigma_2=\sigma^2 I_d$
where $I_d$ is the $d \times d$ identity matrix. Then  \begin{equation}
\frac{W_{(2,\beta)} }{W_{(1,\beta)}}\approx  \sigma^{d(1-\beta)},
\label{eq:dimdegrard} \end{equation} so the ratio of the weights
degenerates exponentially fast in the dimensionality of the problem for
a fixed $\beta$. This heuristic result in \eqref{eq:dimdegrard} shows
that between levels there can be a ``phase-transition'' in the location
of probability mass, which becomes critical as dimensionality increases.

\subsection{Weight Stabilised Gaussian Mixture Targets}
\label{subsec:theidealfix}

Consider targeting a  Gaussian mixture given by 
\begin{equation}
 \pi(x) \propto \sum_{j=1}^J w_j \phi(x,\mu_j,\Sigma_j) \label{eq:Gaussforid}
\end{equation}
in the (practically unrealistic) setting where the target is a Gaussian mixture with known parameters, including the weights. By only tempering the variance component of the modes, one can spread out the modes whilst preserving the component weights. 
We formalise this notion as follows:

\begin{definition}[Weight-Stabilised Gaussian Mixture (WSGM)]
For a Gaussian mixture target distribution $\pi(\cdot)$, as in \eqref{eq:Gaussforid}, the weight-stabilised Gaussian mixture (WSGM) target at inverse temperature level $\beta$ is defined as 
\begin{equation}
\pi_\beta^{WS}(x) \propto \sum_{j=1}^J w_j \phi\left(x,\mu_j,\frac{\Sigma_j}{\beta}\right). \label{dfn:Idealtar}
\end{equation}
\end{definition}

Figure~\ref{Fig:idealweight} shows the comparison between the target distributions used when using power-based targets vs the WSGM targets for the example from Figure~\ref{Fig:badweight}.

\begin{figure}[ht]
\begin{center}
\includegraphics[keepaspectratio,width=9cm]{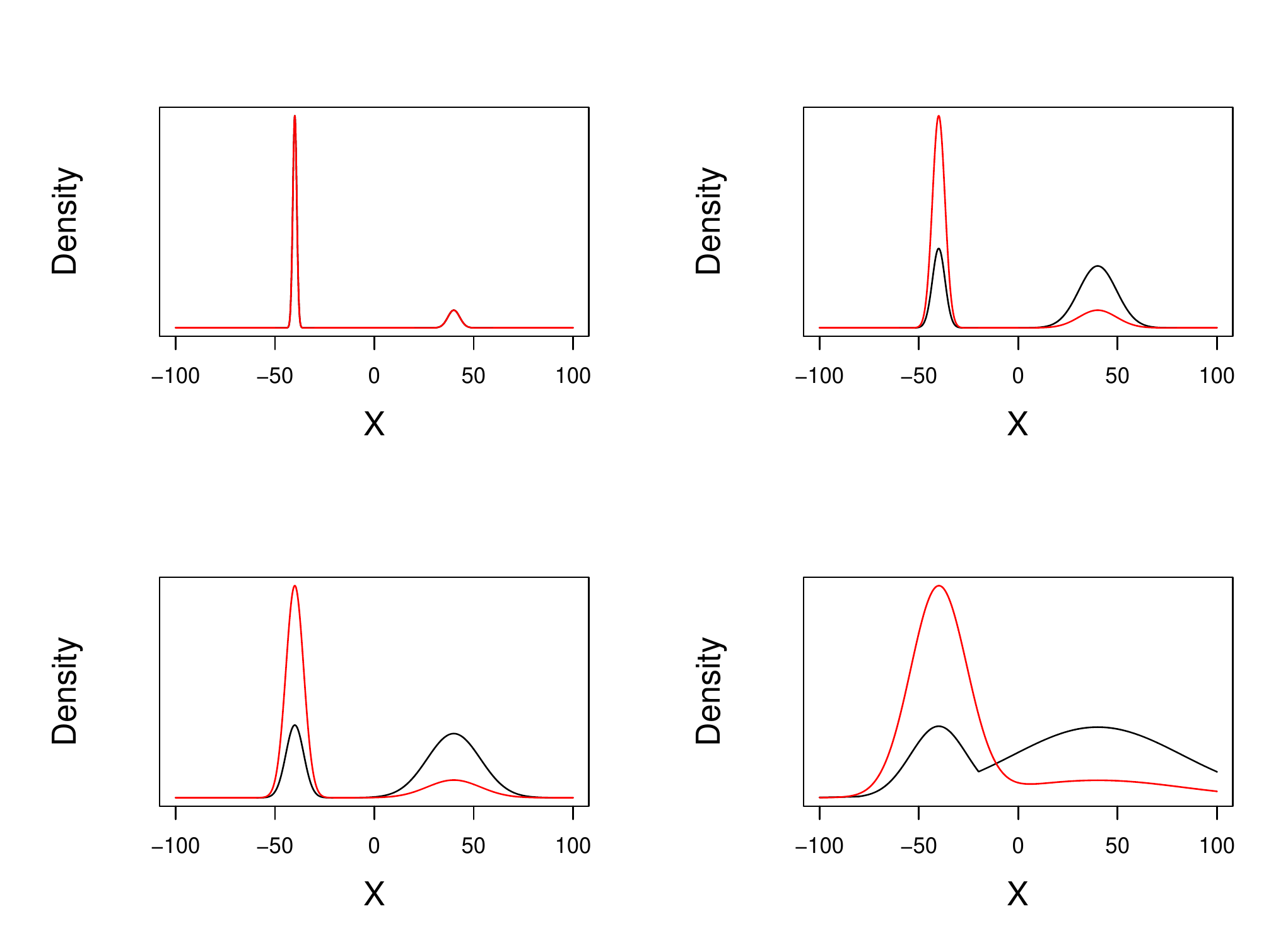}
\caption{For a bimodal Gaussian mixture target, plots of the normalised tempered target densities generated by both powering the target (black) and the WSGM targets (red) at inverse temperature levels $\beta=\{ 1,0.1,0.05,0.005 \}$.}
\label{Fig:idealweight}
\end{center}
\end{figure}

Using these WSGM targets in the PT scheme can give substantially better performance than when using the standard power based targets. This is very clearly illustrated in the simulation study section below in Section~\ref{subsec:WarningHigh}. Henceforth, when the term ``WSGM ST/PT Algorithm'' is used it refers to the implementation of the standard ST/PT algorithm but now using the WSGM targets from (\ref{dfn:Idealtar}).

\subsection{Approximating the WSGM Targets}
\label{sec:wsgmapprox}

In practice, the actual target distribution would be non-Gaussian, and
only approximated by a Gaussian mixture target.
On the other hand,
due to the improved performance gained from using the WSGM over
just targeting the respective power-tempered mixture, there is motivation
to approximate the WSGM in the practical setting where parameters are
unknown. To this end, we present a theorem establishing useful equivalent
forms of the WSGM; these alternative equivalent forms give rise to a
practically applicable approximation to the WSGM.


To establish the notation, let the target be a mixture distribution given by
\begin{equation}
	\pi(x) \propto \sum_{j=1}^J h_j(x) = \sum_{j=1}^J w_j g_j(x)
\label{eq:themixer}
\end{equation}
where $g_j(x)$ is a normalised target density. Then set
\begin{equation}
	\pi_\beta(x) \propto \sum_{j=1}^J f_j(x,\beta) = \sum_{j=1}^J W_{(j,\beta)} \frac{[g_j(x)]^\beta}{\int [g_j(x)]^\beta dx}. \label{eq:themixerbeta}
\end{equation}
Then we have the following result, proved in the Appendix.

\begin{theorem}[WSGM Equivalences] 
Consider the setting of \eqref{eq:themixer} and \eqref{eq:themixerbeta}  where the mixture components consist of multivariate Gaussian distributions i.e.\ $g_j(x) = \phi(x;\mu_j,\Sigma_j)$. Then $\forall j \in 1,\ldots, J$
\begin{enumerate}[(a)]
	\item
\relax [Standard, non-weight-preserving tempering] \\
If $f_j(x,\beta) = [h_j(x)]^\beta$ then 
	\begin{equation}
		W_{(j,\beta)} \propto w_j^\beta |\Sigma_j|^{\frac{1-\beta}{2}}.
		\nonumber
	\end{equation}
        \item
\relax [Weight-preserving tempering, version~\#1] \\
Denoting $\nabla_j = \nabla\log{h_j(x)}$ and $\nabla_j^2=\nabla^2\log{h_j(x)}$;  if $f_j(x,\beta)$ takes the form
\begin{equation}
  h_j(x)\exp\left\{  \left(\frac{1-\beta}{2}\right) (\nabla_j(x))^T  \left[\nabla^2_j(x)\right]^{-1} \nabla_j (x)  \right\}. \nonumber
\end{equation}
 then $W_{(j,\beta)} \propto w_j$.
       \item
\relax [Weight-preserving tempering, version~\#2] \\
If 
\begin{equation}
	f_j(x,\beta)= h_j(x)^\beta h_j(\mu_j)^{(1-\beta)} \nonumber
\end{equation}
then  $W_{(j,\beta)} \propto w_j$.
\end{enumerate}
\label{Theorem:equivalence}
\end{theorem}


\begin{flushleft}
\textbf{\textit{Remark 1:}} In Theorem~\ref{Theorem:equivalence},
statement (b) shows that second order gradient information of the $h_j$'s
can be used to preserve the component weight in this setting.

\textbf{\textit{Remark 2:}} Statement (c) extends statement (b) to no
longer require the gradient information about the $h_j$ but simply the
mode/mean point $\mu_j$. Essentially this shows that by appropriately
rescaling according to the height of the component as the components are
``powered up'' then component weights are preserved in this setting.

\textbf{\textit{Remark 3:}} A simple calculation shows that statement
(c) holds for a more general mixture setting when all components of the
mixture share a common distribution but  different location and scale
parameters.
\end{flushleft}

\subsection{Hessian Adjusted Tempering}

The results of Theorem~\ref{Theorem:equivalence} are derived under the
impractical setting that the components are all known and that
$\pi(\cdot)$ is indeed a mixture target. One would like to exploit the
results of (b) and (c) from Theorem~\ref{Theorem:equivalence} to aid
mixing in a practical setting where the target form is unknown but may
be well approximated by a mixture.

Suppose herein that the modes of the multimodal target of interest,
$\pi(\cdot)$, are well separated. Thus an approximating mixture of the
form given in \eqref{eq:themixer} would approximately satisfy\[ \pi(x)
\propto h_M(x)\] where $M=\sup_j\left\{ h_j(x) \right\}$.
Hence it is tempting to apply a version of
Theorem~\ref{Theorem:equivalence}(b)
to $\pi(\cdot)$ directly as opposed to the $h_j$. So at
inverse temperature $\beta$, the point-wise target would be proportional to
\begin{equation}
	\pi(x)\exp\left\{  \left(\frac{1-\beta}{2}\right) (\nabla_\pi(x))^T  \left[\nabla^2 _\pi(x)\right]^{-1} \nabla_\pi(x)  \right\}. \nonumber
\end{equation}
where $\nabla_\pi=\nabla\log{\pi(x)}$ and $\nabla^2_\pi=\nabla^2 \left(\log{\pi(x)}\right)$.
This only uses point-wise gradient information up to second order. At
many locations in the state space, provided that $\beta$ is at a temperature level that is sufficiently cool that the tail overlap is insignificant, and if the target was indeed a Gaussian mixture then this approach would give almost exactly the
same evaluations as $\pi_\beta(\cdot)$ from \eqref{eq:themixerbeta} in the
setting of (b). However, at locations between modes when the Hessian of
$\log(\pi(x))$ is positive semi-definite, this target behaves very
badly, with explosion points that make it improper.

Instead, under the setting of well separated modes then one can appeal
instead to the weight preserving characterisation in
Theorem~\ref{Theorem:equivalence}(c). Assume that one can assign each
location in the state space to a ``mode point'' via some function $x
\rightarrow \mu_{x,\beta}$, with a corresponding tempered target given by
\[\pi_{\beta}(x)  \propto  \pi(x)^{\beta} \pi(\mu_{x,\beta})^{1-\beta}.\] Note the mode assignment function's dependence on $\beta$. This can be understood to be necessary by appealing to Figure~\ref{Fig:idealweight} where it is clear that the narrow mode in the WSGM target  has a  ``basin of attraction'' that expands as the temperature increases.

\begin{definition}[Basic Hessian Adjusted Tempering (BHAT) Target] \label{def:robadjB}
For a target distribution $\pi(\cdot)$ on $\mathbb{R}^d$ with a corresponding ``mode point assigning function'' $\mu_{x,\beta}: \mathbb{R}^d \rightarrow \mathbb{R}^d$; the BHAT target at inverse temperature level $\beta \in (0,\infty)$ is defined as
\begin{equation}
	\pi_{\beta}^{BH}(x)  \propto  \pi(x)^{\beta} \pi(\mu_{x,\beta})^{1-\beta}. \label{eq:robadj1B}
\end{equation}
\end{definition}

However, in this basic form there is an issue with this target distribution at hot temperatures when $\beta \rightarrow 0$. The problem is that it leaves discontinuities that can grow exponentially large and this can make the hot state temperature level mixing exponentially slow if using standard MCMC methods for the within temperature moves.

This problem can be mitigated if one assumes more knowledge about the target distribution. Suppose that the mode points are known and so there is a collection of $K$ mode points $M=\{ \mu_1,\ldots,\mu_K  \}$. This assumption seems quite strong but in general if one cannot find mode points then this is essentially saying that one cannot find the basins of attraction and thus the desire to obtain the modal relative masses (as MCMC is trying to do) must be relatively impossible. Indeed, being able to find mode points either prior to or online in the run of the algorithm is possible e.g.\ \cite{Tjelmeland2001}, \cite{Behrens2008} and \cite{ALPSTawn}. Furthermore, assume that the target, $\pi(\cdot)$, is $C^2$ in a neighbourhood of the $K$ mode locations and so there is an associated collection of positive definite covariance matrices $S=\{ \Sigma_1,\ldots,\Sigma_K \}$ where $\Sigma_j= -\left(\nabla^2 \log \pi(\mu_j) \right)^{-1}$. From this and knowing the evaluations of $\pi(\cdot)$ at the mode points then one can approximate the weights in the regions to attain a collection $\mathbf{\hat{W}}=\{\hat{w}_1,\ldots,\hat{w}_K\}$ where
\[
\hat{w_j}= \frac{\pi(\mu_j) |\Sigma_j|^{1/2}}{\sum_{k=1}^K \pi(\mu_k) |\Sigma_k|^{1/2}}
\]

With $\phi(\cdot|\mu_j,\Sigma_j)$ denoting the pdf of a $N(\mu_j,\Sigma_j)$ then we define the following modal assignment function motivated by the WSGM:
\begin{definition}[WSGM mode assignment function] \label{def:modeassign}
With collections $M$, $S$ and $\mathbf{\hat{W}}$ specified above then for a location $x\in \mathbb{R}^d$ and inverse temperature $\beta$ define the WSGM mode assignment function as 
\begin{equation}
	A(x,\beta) = \argmax_{j}\left\{ \hat{w}_j \phi \left(x|\mu_j,\frac{\Sigma_j}{\beta}\right)\right\}. \label{eq:ass}
\end{equation}
\end{definition}

Under the assumption that there are collections  $M$, $S$ and $\mathbf{\hat{W}}$ that have either been found through prior optimisation or through an adaptive online approach we define the following:

\begin{definition}[Hessian Adjusted Tempering (HAT) Target] \label{def:robadj}
For a target distribution $\pi(\cdot)$ on $\mathbb{R}^d$ with collections $M$, $S$ and $\hat{W}$ defined above along with the associated mode assignment function given in \eqref{eq:ass}, then the 
Hessian adjusted tempering (HAT) target is defined as
\begin{equation}
	\pi_{\beta}^{H}(x)  \propto
  \begin{cases}
    \pi(x)^{\beta} \pi(\mu_{A(x,\beta)})^{1-\beta} & \text{if $A(x,\beta)=A(x,1) $} \\
    G(x,\beta) & \text{if $A(x,\beta) \neq A(x,1) $}
  \end{cases}  \label{eq:robadj1}
\end{equation}
where with $\hat{A}:=A(x,\beta)$
\[
G(x,\beta)= \frac{\pi(\mu_{\hat{A}})\left((2\pi)^{d}\Sigma_{\hat{A}}\right)^{1/2}\phi \left(x|\mu_{\hat{A}},\frac{\Sigma_{\hat{A}}}{\beta}\right)}{\beta^{d/2}}.
\]
\end{definition}

Essentially the function ``$G$'' specifies the target distribution when the chain's location, $x$, is in a part of the state space where the narrower modes expand their basins of attraction as the temperature gets hotter. Both the choice of $G$ and the mode assignment function used   in  Definition~\ref{def:robadj} are  somewhat canonical to the Gaussian mixture setting. With the same assignment function specified in Definition~\ref{def:modeassign}, an alternative and seemingly robust ``$G$'' that one could use is given by
\begin{eqnarray}
&& G(x,\beta)= \pi(x,1,A) \nonumber \\
&& + \left(\frac{2P(A(x,\beta))}{P(A(x,\beta))+P(A(x,1))}-1\right)[\pi(x,\beta,A) -\pi(x,1,A) ] \nonumber
\end{eqnarray}
where $\pi(x,\beta,A)=  \pi(x)^{\beta} \pi(\mu_{A(x,\beta)})^{1-\beta}$ and $P(j)=\hat{w}_j \phi \left(x|\mu_j,\frac{\Sigma_j}{\beta}\right)$.

With either of the suggested forms of the function $G$ then under the assumption that
the target is continuous and bounded on $\mathbb{R}^d$,
and that for all $\beta \in (0,\infty)$,
\[  \int_{\mathcal{X}} \pi^\beta(x) d x <\infty \, ,\]
then  $\pi_{\beta}^{H}(x)$
is a well defined probability density, i.e.\
Definition~\ref{def:robadj} makes sense.

For a bimodal Gaussian mixture example Figure~\ref{fig:HATillu} compares the HAT target relative to the WSGM target; showing that the HAT targets are a very good approximation to the WSGM targets, even at the hotter temperature levels.
\begin{figure}[ht]
\begin{center}
\includegraphics[keepaspectratio,width=8cm]{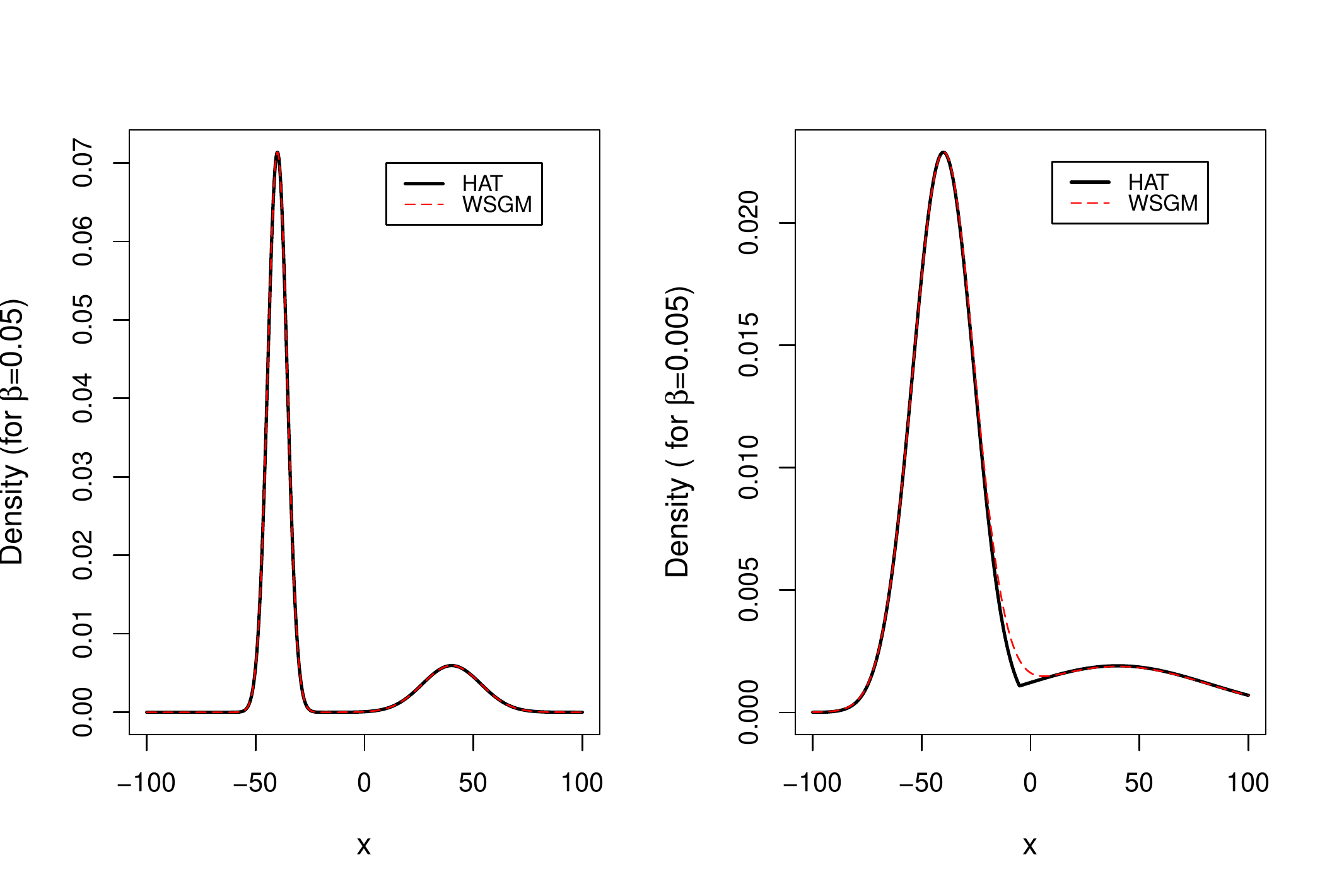}
\caption{For the same bimodal Gaussian target from Figure~\ref{Fig:idealweight}, here is a comparison of the HAT vs WSGM targets at inverse temperatures $\beta=0.05$ and $\beta=0.005$ respectively. Note they are almost identical at the colder temperature; but they do differ slightly in the interval $(-25 , 10)$ at the hotter temperature where the ``G'' function is allowing the footprint of the narrow mode to expand.}
\label{fig:HATillu}
\end{center}
\end{figure}


We propose to use the HAT targets in place of the
power-based targets for the tempering algorithms given in
Section~\ref{sec:parallel}.  We thus define the following algorithms,
which are explored in the following sections.

\medskip
\begin{definition}[Hessian Adjusted Simulated Tempering (HAST) Algorithm]
The HAST algorithm is an implementation of the ST algorithm
(Section~\ref{sec:parallel}, Algorithm~\ref{alg:ST}) where the target
distribution at inverse temperature level $\beta$ is given by
$\pi_{\beta}^{H}(\cdot)$ from Definition~\ref{def:robadj}.
\label{HASTdef}
\end{definition}

\medskip
\begin{definition}[Hessian Adjusted (Parallel) Tempering (HAT) Algorithm]
The HAT algorithm is an implementation of the PT algorithm
(Section~\ref{sec:parallel}, Algorithm~\ref{alg:PT}) where the target
distribution at inverse temperature level $\beta$ is given by
$\pi_{\beta}^{H}(\cdot)$ from Definition~\ref{def:robadj}.
\label{HATdef}
\end{definition}

\section{Simulation Studies}
\label{sec:simstudy}

\subsection{WSGM Algorithm Simulation Study}
\label{subsec:WarningHigh}

We begin by comparing the performances of a ST algorithm targeting both
the power-based and WSGM targets for a simple  but challenging bimodal Gaussian mixture  target
example.  The example will illustrate that the traditional ST algorithm,
using power-based targets, struggles to mix effectively through the
temperature levels due to a bottleneck effect caused by the lack of
regional weight preservation.

The example considered is the 10-dimensional target distribution given by the bi-modal Gaussian mixture
\begin{equation}
\pi(x)= w_1 \phi_{(\mu_1,\Sigma_1)}(x)+w_2 \phi_{(\mu_2,\Sigma_2)}(x)
\end{equation}
where $w_1=0.2$, $w_2=0.8$, $\mu_1=(-10,-10,\ldots,-10)$, $\mu_2=(10,10,\ldots,10)$, $\Sigma_1= 9\mbox{\bf{I}}_{10}$ and $\Sigma_2= \mbox{\bf{I}}_{10}$.
When power based tempering is used, then
mode~1 accounts for only 20\% of the mass at the cold level,
but at the hotter temperature levels
becomes the dominant mode containing almost all the mass.

For both runs the same geometric temperature schedule was used: \[\Delta=\{1, \, 0.32, \, 0.32^2, \, \ldots, \, 0.32^6  \}.\] This geometric schedule is justified by Corollary~1 of \cite{Tawn2018rwopt}, which suggests this is an optimal setup in the case of a regionally weight preserved PT setting. Indeed, this schedule induces a swap move acceptance rates around 0.22 for the WSGM algorithm; close to the suggested 0.234 optimal value. 

This temperature schedule gave swap acceptance rates of approximately 0.23 between all levels of the power-based ST algorithm except for the coldest level swap where this degenerated to 0.17. That shows that the power-based ST algorithm was set up essentially optimally according to the results in \cite{atchade2011towards}.


In order to ensure that the within-mode mixing isn't influencing the
temperature space mixing, a local modal independence sampler was
used for the within-mode moves. This essentially means that once a mode has
been found, the mixing is infinitely fast. We use the
modal assignment function $\mu_{x,\beta}$ which specifies that the
location $x$ is in
mode 1 if $\bar{x}<0$ and in mode 2 otherwise.
Then the within-move proposal distribution for a move at
inverse temperature level $\beta$ is given by
\begin{equation}
	q_{\beta}(x,y)=\phi_{\left(\mu_1,\frac{\Sigma_1}{\beta}\right)}(y)\mathbbm{1}_{\bar{x}<0}+\phi_{\left(\mu_2,\frac{\Sigma_2}{\beta}\right)}(y)\mathbbm{1}_{\bar{x}\ge 0},
\end{equation}
where $\phi_{\mu,\Sigma}(.)$ is the density function of a Gaussian random variable with mean $\mu$ and variance matrix $\Sigma$.

Figure~\ref{Fig:Trace10DHardIdeal} plots a functional of the inverse temperature at each iteration of the algorithm. The functional is
\begin{equation}
	h(\beta_t,x_t)
:=\frac{\log \left(  \frac{\beta_t}{\beta_{\text{min}}} \right)}{\log \left(\frac{1}{\beta_{\text{min}}}\right)}\mbox{sgn}\left( \bar{x}_t \right) \label{eq:tracefun}
\end{equation}
where $\mbox{sgn}(.)$ is the usual sign function and $\beta_{\text{min}}$ is the minimum of the inverse temperatures. The significance of this functional will become evident from the results of the core theoretical contributions made in this paper in Theorems~\ref{Theorem:diffusion} and \ref{Theorem:skewBM} in Section~\ref{sec:Theoreticalanalysis}. Essentially it is taking a transformation of the current inverse temperature at iteration $t$ of the Markov chain, such that when $\beta_t=1$ the magnitude of $h$ is 1 and when the temperature is at its hottest level, i.e.\ $\beta_t = \beta_{\text{min}}$, then $h$ is zero. Furthermore, in this example the sign of $\bar{x}_t$ is a reasonable proxy to identify the mode that the chain is contained in with a negative value suggesting the chain is in the mode centred on $\mu_1$ and $\mu_2$ otherwise.

Figure~\ref{Fig:Trace10DHardIdeal} clearly illustrates that the hot state modal weight inconsistency leads the chain down a poor trajectory since
at hot temperatures nearly all the mass is in modal region 1.
This results in the chain never reaching the other mode in the entire
(finite) run of the algorithm.
Indeed, the trace plots in Figure~\ref{Fig:Trace10DHardIdeal} show that
the chain is effectively trapped in mode 1, which although it only has
20\% of the mass in the cold state, is completely dominant at the hotter states.

\begin{figure}[ht]
\begin{center}
\includegraphics[keepaspectratio,width=8cm]{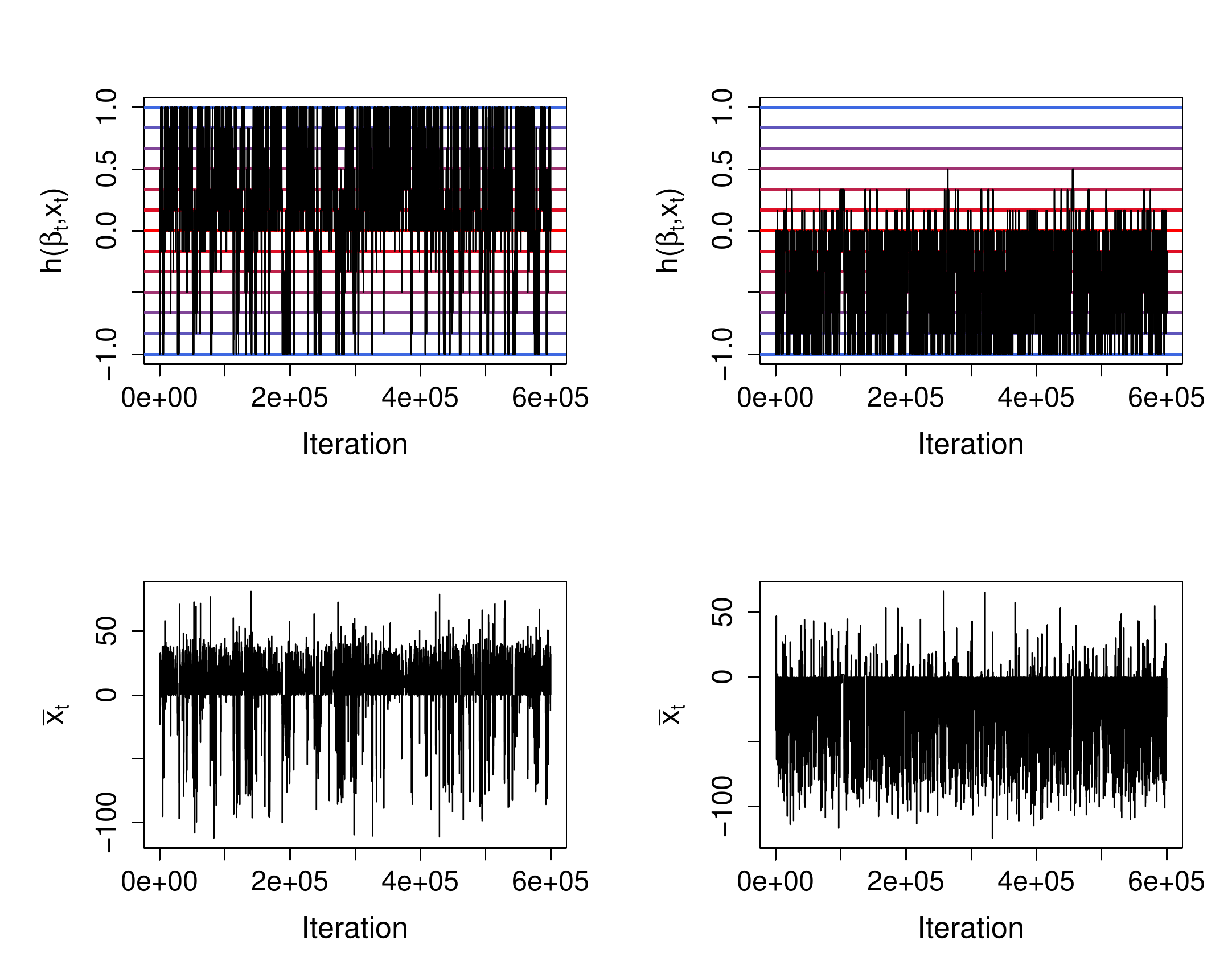}
\caption{Top: Trace plots of the functional of the simulated tempering chains given in (\ref{eq:tracefun}). On the left is the version using the WSGM targets, which mixes well through the temperature schedule and finds both modal regions. On the right is the version using the standard power-based targets, which fails to ever find one of the modes. Bottom: Trace plots of $\bar{x}_t$ in each of the cases respectively.}
\label{Fig:Trace10DHardIdeal}
\end{center}
\end{figure}

\subsection{Simulation study for HAT}
\label{sec:HATExamples}

To demonstrate the capabilities of the HAT algorithm in a non-Gaussian setting where the modes exhibit skew then a five-dimensional four-mode skew-normal mixture target example is presented. Albeit a mixture, this example can be seen to give similar target distribution  geometries to non-mixture settings due to the skew nature of the modes.
\begin{equation}
 	\pi(x) \propto \sum_{k=1}^4 w_k \prod_{i=1}^5 f(x_i|\mu_k,\sigma_k,\alpha) \label{eq:5dimex}
\end{equation}
where the skew normal density is given by 
\[
f(z|\mu,\sigma,\alpha)= \frac{2}{\sigma}\phi\left( \frac{z-\mu}{\sigma}\right)\Phi\left( \frac{\alpha(z-\mu)}{\sigma}\right)
\] 
and where $w_1=w_2=w_3=w_4=0.25$, $\mu_1=-15$, $\mu_2=15$, $\mu_3=45$, $\mu_4=-45$, $\sigma_1=1$, $\sigma_2=1$,  $\sigma_3=3$,  $\sigma_4=3$ and $\alpha=2$.

As will be seen in the forthcoming simulation results the imbalance of scales within each modal region ensures that this is a very challenging problem for the PT algorithm.
 
Since this target fits into the setting of Corollary~1 of \cite{Tawn2018rwopt} then a geometric inverse temperature schedule is approximately optimal for the HAT target in this setting. Indeed,  \cite{Tawn2018rwopt} suggest that the geometric ratio should be tuned so that the acceptance rate for swap moves between consecutive temperatures is approximately 0.234. In this case, eight tempering levels were used to obtain effective mixing; these were geometrically spaced and given by $\{1,0.31,0.31^2,\ldots, 0.31^7  \}$, was found to be approximately optimal and gave an average of 0.22 for the swaps between consecutive levels for the HAT algorithm. 

Using  this temperature schedule along with appropriately tuned RWM proposals for the within temperature  moves,  10 runs of both the PT and HAT algorithms were performed. In each individual run,  each temperature marginal was updated with $m=5$ RWM proposals followed by a temperature swap move proposal and this was repeated with $s=100,000$ sweeps. This results in a sample output of 600,001 of the cold state chain prior to any burn-in removal. Herein for this example denote  $N=600,001$.

As expected, the scale imbalance between the modes resulted in the PT algorithm performing poorly and with significant bias in the sample output. In contrast, the HAT
approach was highly successful in converging relatively rapidly to the target distribution, exhibiting far more frequent inter-modal jumps at the cold state.

Figure~\ref{fig:FiveIdeal} shows one representative example  of a run of the PT and HAT algorithms by plotting the first component of the five-dimensional marginal chain at the coldest target state. It illustrates the impressive inter-modal mixing of HAT across all 4 modal regions as opposed to the very sticky mixing exhibited by the PT algorithm. 

\begin{figure}[ht]
\begin{center}
\includegraphics[keepaspectratio,width=8.5cm]{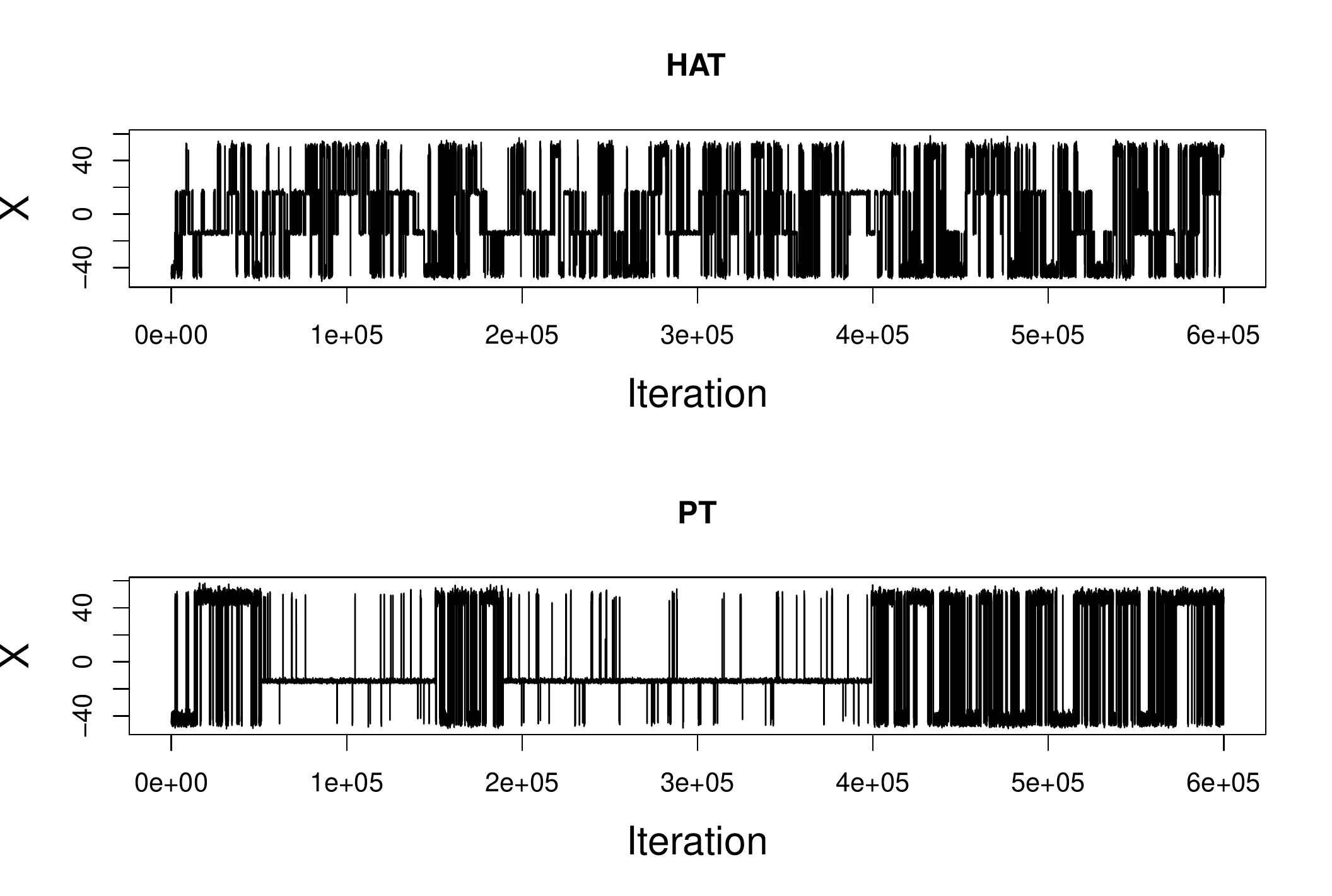}
\caption{Two trace plots of the first marginal component cold state chain targeting the distribution in (\ref{eq:5dimex}) using the HAT and PT algorithms respectively. Note the HAT algorithm run illustrates a chain that is performing rapid exploration between all four modes whereas the PT algorithm exhibits significant sticky patches.}
\label{fig:FiveIdeal}
\end{center}
\end{figure}

Figure~\ref{fig:fiveDmovingweightapproximation} shows the running approximation of $\mathbb{P}_{\pi}(-30<X^1_i<0)$ (which is approximately the weight of the first mode i.e.\ $w_1=0.25$) after the $k^{th}$ iteration of the cold
state chains, after removing a burn-in period of 10,000 initial iterations,
for the ten runs of the PT and HAT runs respectively. The
 approximation after iteration $k \le N$ is given by
\begin{equation}
	\hat{W}^k_1 :=  \frac{1}{k-10000}\sum_{i=10001}^k  \mathbbm{1}_{\left(-30<X^1_i<0\right)} \label{eq:weighestimator}
\end{equation}
where $X^1_i$ is the location of the first component of the five-dimensional chain at the coldest temperature level after the $i^{th}$ iteration.
This figure indicates that the PT algorithm fails
to provide a stable estimate for $\mathbb{P}_{\pi}(-30<X^1_i<0)$ with the running weight approximations far from stable at the end of the runs; in stark contrast  the HAT algorithm exhibits very stable performance in this case. In fact the final estimates for $\hat{W}^N_1$ the are given in Table~\ref{tab:tib1}.

\begin{figure}[ht]
\begin{center}
\includegraphics[keepaspectratio,width=9cm]{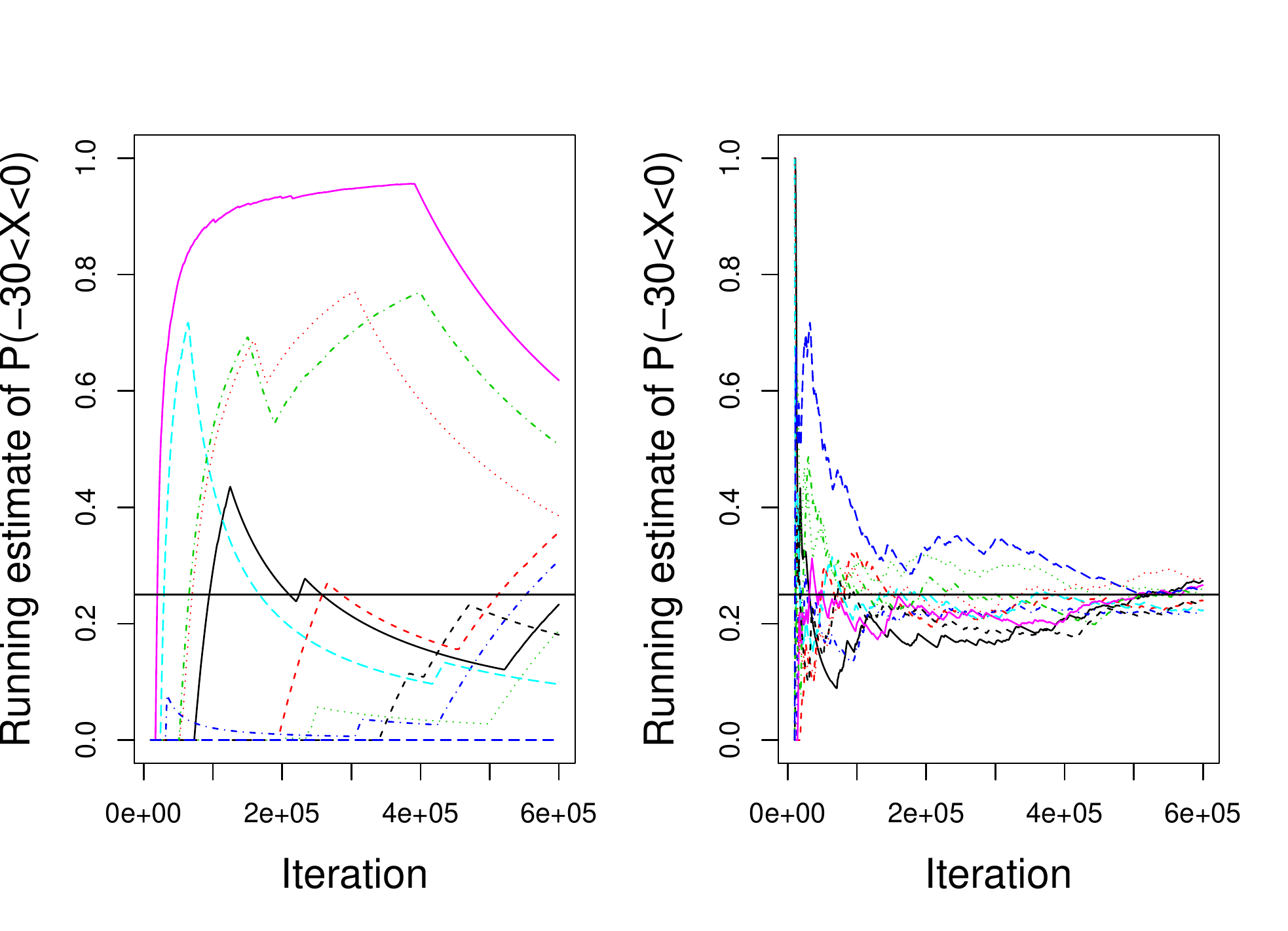}
\caption{Running estimate of  $\mathbb{P}_{\pi}(-30<X^1_i<0)$, i.e.\ $\hat{W}^k_1$ given in (\ref{eq:weighestimator}), for 10 runs of the PT (left) and HAT (right) algorithms. The horizontal black  line at the level 0.25 represents the true probability that one aims to target in this case. In both cases a burn-in of 10,000 iterations was removed. Observe the lack of convergence of the weight estimates for the PT runs compared to the relatively impressive estimates from the HAT runs.}
\label{fig:fiveDmovingweightapproximation}
\end{center}
\end{figure}

\begin{table}
\caption{The end point estimates, $\hat{W}^N_1$, of $\mathbb{P}_{\pi}(-30<X^1_i<0)$ from the 10 runs of the PT and HAT algorithms. The true value of 0.25 appears to be well approximated by HAT but not by PT.}
\label{tab:tib1}       
\begin{tabular}{l|llllllllll}
\hline\noalign{\smallskip}
PT & 0.23 & 0.36 & 0.19 & 0.31 & 0.10 & 0.12 & 0.18 \\

& 0.39 &0.51 & 0 \\
\noalign{\smallskip}\hline\noalign{\smallskip}
HAT & 0.27 & 0.24 & 0.26 & 0.22 & 0.22 & 0.27 & 0.23 \\
& 0.28 & 0.25 & 0.26 \\
\noalign{\smallskip}\hline
\end{tabular}
\end{table}

Table~\ref{tab:tib2} presents the results of using the 10 runs of each algorithm in a batch-means approach to estimate the Monte Carlo variance of the estimator of $\hat{W}^N_1$. The results in Table~\ref{tab:tib2} show that the Monte Carlo error is approximately a factor of 10 higher for the PT algorithm than the HAT approach. 

However, it is also important to analyse this inferential gain jointly with the increase in computational cost. Table~\ref{tab:tib2} also shows that the average run time for the 10 HAT runs was 451 seconds which is a little more than 2 times slower than the average run time of the PT algorithm (217 seconds) in this example. The major extra expense is due to the cost of computing the WSGM mode assignment function in \eqref{eq:ass} at both the cold and current temperature of interest at each evaluation of the HAT target. Anyhow, this is very promising since for a little more than twice the computational cost the inferential accuracy appears to be ten times better in this instance.

\begin{table}
\caption{Using the 10 runs of each algorithm in a batch-means approach to estimate the Monte Carlo variance of the pooled estimator $\hat{W}^{10N}_1$ i.e.\ $\text{SD}(\hat{W}^N_1)$. Also displayed is the average run time (RT, measured in seconds) of a single one of the 10 repaeted runs for both methods respectively.}
\label{tab:tib2}     
\begin{tabular}{l|llll}
\hline\noalign{\smallskip}
 & $\hat{W}^{10 N}_1$  & $\hat{\text{SD}}(\hat{W}^N_1)$ & $\hat{\text{SD}}$($\hat{W}^{10 N}_1$) & RT (secs)\\
\noalign{\smallskip}\hline\noalign{\smallskip}
PT & 0.288 & 0.187 & 0.0593 & 217 \\
\noalign{\smallskip}\hline\noalign{\smallskip}
HAT &  0.249 & 0.019 & 0.0063 & 451 \\
\noalign{\smallskip}\hline
\end{tabular}
\end{table}

\section{Diffusion Limit and Computational Complexity}
\label{sec:Theoreticalanalysis}

In this section, we provide some theoretical analysis for our algorithm.
We shall prove in Theorems~\ref{Theorem:diffusion}
and~\ref{Theorem:skewBM} below that as the
dimension goes to infinity, a simplified and speeded-up version of our
weight-preserving simulated tempering algorithm (i.e., 
the HAST Algorithm from Definition~\ref{HASTdef}, equivalent to
the ST Algorithm~\ref{alg:ST} with the adjusted target from
Definition~\ref{def:robadj}) converges to a certain
specific diffusion limit.  This limit will allow us to draw some
conclusions about the computational complexity of our algorithm.


\subsection{Assumptions}

We assume for simplicity (though see below) that
our target density $\pi$ is a mixture of the form~\eqref{eq:themixer}
with just $J=2$ modes, of weights $w_1=p$ and $w_2=1-p$ respectively,
with each mixture component a special i.i.d.\ product
$g_j(x) = \prod_{i=1}^d
f_j(x_i)
$ as in \eqref{eq:restrict}.
We further assume that a weight-preserving transformation
(perhaps inspired by Theorem~\ref{Theorem:equivalence}(b) or~(c))
has already been done, so that
\begin{eqnarray*}
\pi_\beta(x)
\ &\propto& \ p {[g_1(x)]^\beta \over \int [g_1(x)]^\beta dx}
+ (1-p) {[g_2(x)]^\beta \over \int [g_2(x)]^\beta dx} \\
 \\ &\equiv& \ p g_1^\beta(x) + (1-p) g_2^\beta(x)
\end{eqnarray*}
for each $\beta$.
We consider a simplified version of the weight-preserving process,
in which the chain always mixes immediately within each mode, but
the chain can only jump between modes when at the hottest temperature $\beta_{min}$, at which point it jumps to
one of the two modes with probabilities $p$ and $1-p$ respectively.
Let $I$ denote the indicator of which mode the process is in, taking value $1$ or $2$.

We shall sometimes concentrate on the {\it Exponential Power Family} special case in which each of the two mixture component factors is of the form $f_j(x) \propto e^{-\lambda_j|x|^{r_j}}$ for some
$\lambda_j,r_j>0$.  This includes the Gaussian case for which
$r_1=r_2=2$ and $\lambda_j = 1/\sigma_j^2$.
(Note that the HAT target in \eqref{eq:robadj1}  requires the existence of second derivatives about the mode points, corresponding to $r_j \ge 2$.)

As in \cite{atchade2011towards} and \cite{roberts2014minimising},
following \cite{predescu2004incomplete} and \cite{kone2005selection},
we assume that the inverse temperatures are given by
$1=\beta_0^{(d)},\beta_1^{(d)},
\ldots,\beta^{(d)}_{k(d)} \approx \beta_{\text{min}}$,
with
\begin{equation}
\label{eqn:beta}
\beta_i = \beta_{i-1} -
\ell(\beta_{i-1})/d^{1/2}
\end{equation}
for some fixed $C^1$ function $\ell$.
In many cases, including the Exponential Power Family case,
the optimal choice of $\ell$ is
$\ell(\beta) = \beta \ell_0$ for a constant $\ell_0 \doteq 2.38$.



We let $\beta_t^{(d)}$ be the inverse temperature at time $t$
for the $d$-dimensional process.
To study weak convergence, we let
$ \beta^{(d)}_{N(dt)}$ be a continuous-time version of the
$\beta_t^{(d)}$ process, speeded up by a factor of $d$, where $\{N(t)\}$
is an independent standard rate~1 Poisson process.
To combine the two modes into one single process, we further augment
this process by multiplying it by $-1$ when the algorithm's state
is closer to the second mode,
while leaving it positive (unchanged) when state is closer
to the first mode. Thus define
\begin{equation}
X_t^{(d)}= (3 - 2I) \, \beta^{(d)}_{N(dt)}
\, .   \label{eq:xtddef}
\end{equation}



\subsection{Main Results}

Our first diffusion limit result (proved in the Appendix), following
\cite{roberts2014minimising}, states that when we are at an inverse
temperature greater than $\beta_{\text{min}}$, the inverse temperature process
behaves identically to the case where there is only one mode (i.e.\ $J=1$).

\begin{theorem}\label{Theorem:diffusion}
Assume the target  $\pi$ is of the form~\eqref{eq:themixer},
with $J=2$ modes of weights $w_1=p$ and $w_2=1-p$, with inverse
weights chosen as in~\eqref{eqn:beta}.
Then up until the first time the process $X^{(d)}$ hits $\pm \beta_{\text{min}}$,
as $d\to\infty$, 
$\{X_t^{(d)}\}$ converges weakly 
to a fixed diffusion process $X$ given by \eqref{eq:xtddef}.
\end{theorem}

Theorem \ref{Theorem:diffusion} described what happens away from $\beta_{min}$. However it says nothing about what happens at
$\beta_{\text{min}}$. Moreover its statespace $[-1,-\beta_{\text{min}})\cup
(\beta_{\text{min}}, 1]$
 is not connected, and we have not even properly defined $h$ at 
$\pm \beta _{\text{min}}$.  To resolve these issues we define
$$
h(x) = \begin{cases}
\int_{\beta _{\text{min}}}^x {1\over \ell (u)} du, &\hbox{ when }  x>0\\
-\int_{\beta _{\text{min}}}^{-x}
{1\over \ell (u)} du,  &\hbox{ when }  x<0\\
0,  &\hbox{ when }  x=0\\
\end{cases}
$$
and set $H_t= h(X_t)$, thus making the process $H$ continuous at~0.

\begin{remark}
\label{Remark:Uniform}
The process $H$
 
leaves constant densities locally invariant, ${\tilde G}^*g(v)=0$ for all
$v\neq 0$ where ${\tilde G}^*$ is the adjoint of the infinitesimal
generator of $H$, as will be shown in the Appendix.
This suggests that the density of the
invariant distribution of $H$ (if it exists)
should be piecewise uniform, i.e.\ it should be
constant for $v>0$ and also constant for $v<0$ though
these two constants might not be equal.
\end{remark}


To make further progress, we require a {\it proportionality condition}.
Namely, we assume that the quantities corresponding to
$I(\beta)=\text{Var}_{\pi^\beta}\big( \log f(x) \big)$ are proportional to each
other in the two modes.  More precisely,
we extend the definition of $I$
to $I(\beta) = \text{Var}_{x\sim f_1^\beta}(\log f_1(x))$ for $\beta>0$
(corresponding to the first mode),
and $I(\beta) = \text{Var}_{x\sim f_2^{|\beta|}}(\log f_2(x))$ for $\beta<0$
(corresponding to the second mode),
and assume there is a fixed function $I_0:\mathbb{R}_{+}\to\mathbb{R}_{+}$ and positive
constants $r_1$ and $r_2$ such that we have
$I(\beta) = I_0(\beta)/r_1$ for $\beta>0$ (in the first mode), while
$I(\beta) = I_0(|\beta|)/r_2$ for $\beta<0$ (in the second mode).
For example,
it follows from Section~2.4 of \cite{atchade2011towards} that
in the Exponential Power Family case,
$I(\beta) = 1/r_1\beta^2$ for $\beta>0$ and
$I(\beta) = 1/r_2\beta^2$ for $\beta<0$, so that this proportionality
condition holds in that case.

Corresponding to this, we choose
the inverse temperature spacing function as follows (following
\cite{atchade2011towards} and \cite{roberts2014minimising}):
\begin{equation}
\label{eqn:ell}
\ell(\beta) \ = \ I_0^{-1/2}(\beta) \, \ell_0
\end{equation}
for some fixed constant $\ell_0>0$.

To state our next result, we require the notion of {\it skew Brownian
motion}, a generalisation of usual Brownian motion.  Informally, this
is a process that behaves just like a Brownian motion, except that the
sign of each excursion from~0 is chosen using an independent Bernoulli
random variable; for further details and
constructions and discussion see e.g.\ \cite{Lejay2006}.  
We also require the function
$$
z(h)
\ = \ h\, \left[ 2 \, \Phi\left( {- \ell_0 \over 2 \sqrt{r(h)}}
                                \right) \right]^{-1/2}
\, .
$$
where $r(h)=r_1$ for $h>0$ and $r(h)=r_2$ for $h<0$.
We then have the following result (also proved in the Appendix).

\begin{theorem}\label{Theorem:skewBM}
Under the set-up and assumptions of Theorem~\ref{Theorem:diffusion},
assuming the above proportionality condition and the choice~\eqref{eqn:ell},
then as $d\to\infty$, 
the process $\{X_t^{(d)}\}$ converges weakly in the
Skorokhod topology to a limit process $X$.
Furthermore, the limit process has the property that if
$$
Z_t = z\big( h(X_t) \big)
\, ,
$$
then $Z$ is skew Brownian motion $B^*_t$ with reflection at
\begin{equation}
\label{eq:bdsSBM}
(3-2i) \left[ 2 \, \Phi\left( {- \ell_0 \over 2 \sqrt{r_i}}
                                \right) \right]^{-1/2}
\int_{\beta_{\text{min}}}^1{1\over \ell (u) }du, \ \ \ i=1, 2\ .
\end{equation}
\end{theorem}

\begin{remark}
It follows from the proof of
Theorem~\ref{Theorem:skewBM} that the specific version of skew
Brownian motion $B^*_t$ that arises in the limit is one with excursion
weights proportional to
\[a = p
\left[ 2\Phi\left( {- \ell_0 \over 2 \sqrt{r_1}} \right) \right]^{1/2}
\text{and~}
b = (1-p)
\left[ 2 \Phi\left( {- \ell_0 \over 2 \sqrt{r_2}} \right) \right]^{1/2}.\]
That means that the stationary density for $B^*_t$ on the positive and
negative values is proportional to $a$ and $b$ respectively.
This might seem surprising since the limiting weights of the modes
should be equal to $p$ and $1-p$, not proportional to $a$ and $b$
(unless $r_1=r_2$).  The explanation is that
the {\it lengths} of the positive and negative parts of
the domain are given by
$\left[ 2 \, \Phi\left( {- \ell_0 \over 2 \sqrt{r_1}} \right) \right]^{1/2}$
and
$\left[ 2 \, \Phi\left( {- \ell_0 \over 2 \sqrt{r_2}} \right) \right]^{1/2}$
respectively.
Hence, the total stationary mass of the positive and negative parts
-- and hence also the limiting modes weights --
are still $p$ and $1-p$ as they should be.
\end{remark}


\subsection{Complexity Order}

Theorems~\ref{Theorem:diffusion}
and~\ref{Theorem:skewBM} have implications for the computational
complexity of our algorithm.

In Theorem~\ref{Theorem:diffusion}, the limiting diffusion process $H_t$
is a fixed process, not depending on dimension except through the value
of $\beta_{\text{min}}$.  It follows that if $\beta_{\text{min}}$ is kept fixed, then
$H_t$ reaches 0 (and hence mixes modes) in time $O(1)$.  Since $H_t$ is
derived (via $X_t$) from the $\beta_t$ process speeded up by a factor of
$d$, it thus follows that for fixed $\beta_{\text{min}}$, $\beta_t$ reaches
$\beta_{\text{min}}$ (and hence mixes modes) in time $O(d)$.  So, if
$\beta_{\text{min}}$ is kept fixed, then the mixing time of the
weight-preserving tempering algorithm is $O(d)$, which is very fast.
However, this does not take into account the dependence on
$\beta_{\text{min}}$, which might also change as a function of $d$.

Theorem~\ref{Theorem:skewBM} allows for control of the dependence of
mixing time on the values of $\beta_{\text{min}}$.
The limiting skew Brownian motion process $B^*_t$
is a fixed process, not depending on dimension nor on $\beta_{\text{min}}$, with range given by the reflection points in (\ref{eq:bdsSBM}).
It follows that $Z_t$ reaches 0 (and hence mixes modes)
in time of order the square of the total length of the interval,
i.e.\ of order
$$
\left( \sum_{i=1}^2\left[ 2 \, \Phi\left( {- \ell_0 \over 2 \sqrt{r_i}}
                                \right) \right]^{-1/2}
\int_{\beta_{min}}^1{1\over \ell (u) }du \right)^2
$$
In the Exponential Power Family case, this is easily computed to be
$O\big(d \, [\log \beta_{\text{min}}]^2\big)$.

This raises the question of how large $\beta_{\text{min}}$ needs to be, as a
function of dimension $d$.  If
the proposal scaling is optimal for within
each mode at the cold temperature, then the proposal scaling is
$O(d^{-1/2})$.  Then, at an inverse temperature $\beta$,
the proposal scaling is $O((\beta d)^{-1/2})$.
Hence, at an inverse temperature $\beta$, the probability
of jumping from one mode to the other (a distance $O(\sqrt{d})$ away) is
roughly of order $e^{-\beta d^2}$.  This is exponentially small
unless $\beta = O(1/d^2)$.
This indicates that, for our algorithm to perform well, we
need to choose $\beta_{\text{min}} = O(1/d^2)$.
With this choice, the mixing time order becomes 
$$
\left( \sum_{i=1}^2\left[ 2 \, \Phi\left( {- \ell_0 \over 2 \sqrt{r_i}}
                                \right) \right]^{-1/2}
\int_{1/d^2}^1{1\over \ell (u) }du \right)^2
$$

In the Exponential Power Family case, this corresponds to $O\big(d \,
[\log d]^2\big)$.  That is, for the inverse temperature process to hit
$\beta_{\text{min}}$ and hence mix modes, takes $O\big(d \, [\log d]^2\big)$
iterations. This is a fairly modest complexity order, and compares very
favourably to the exponentially large convergence times which arise for
traditional simulated tempering as discussed in
Subsection~\ref{subsec:torpid}.

\subsection{More than Two Modes}

Finally, we note that
for simplicity, the above analysis was all done for just {\it two} modes.
However, a similar analysis works more generally.  Indeed, suppose now
that we have $k$ modes, of general weights $p_1,p_2,\ldots,p_k \ge 0$ with
$\sum_i p_i = 1$.  Then when $\beta$ gets to $\beta_{\text{min}}$, the process
chooses one of the $k$ modes with probability $p_i$.  This corresponds to
$\{Y_t\}$ being replaced by a Brownian motion not on $[-1,1]$, but rather
on a ``star'' shape with $k$ different length-1 line segments all meeting
at the origin (corresponding, in the original scaling, to $\beta_{\text{min}}$),
where each time the Brownian motion hits the origin it chooses one of the
$k$ line segments with probability $p_i$ each.  This process is called
{\it Walsh's Brownian motion}, see e.g.\ \cite{barlow1989}.  (The case
$k=2$ but $p_1\not=1/2$ corresponds to skew Brownian motion as above.)
For this generalised process, a theorem similar to
Theorem~\ref{Theorem:diffusion} can be then stated and proved by similar
methods, leading to the same complexity bound
of $O\big(d \, [\log d]^2\big)$ iterations in the multimodal case as well.

\section{Conclusion and Further Work}
\label{sec:conc}

This article has introduced the HAT algorithm to mitigate the lack of
regional weight preservation in standard power-based tempered targets.
Our simulation studies show promising mixing results, and our theorems
indicate the mixing times can become polynomial rather than exponential
functions of the dimension $d$, and indeed of time
$O(d[\log d]^2)$ under appropriate assumptions.

Various questions remain to make our HAT approach
more practically applicable. The ``modal assignment function'' needs to be specified in an appropriate way, and more exploration into the robustness of the current assignment mechanism is needed to understand its performance on heavier and lighter tailed distributions. The suggested HAT target assumes knowledge of the mode points which typically one will not have to begin with and one would rely on effective  optimisation methods to seek these out either during or prior to the run of the algorithm. Indeed, this has been partially explored by the authors in \cite{ALPSTawn}. The performance of HAT is heavily reliant on the mixing at the hottest temperature level; the use of RWM here can be problematic for HAT where the mode heights of the disperse modes can be far lower than the narrower modes. As such more advanced sampling schemes such as discretised tempered Langevin could give accelerated mixing at the hot state; the effects of which would be transferred to an improvement in the mixing at the coldest state.

In the theoretical analysis of Section~\ref{sec:Theoreticalanalysis},
the spacing between consecutive inverse-temperature levels was taken to
be $O(d^{-1/2})$ to induce a non trivial diffusion limit. However, this
result required strong assumptions.  Accompanying work in
\cite{Tawn2018rwopt} suggests that for the
HAT algorithm under more general conditions, the consecutive optimal
spacing should still be $O(d^{-1/2})$, with an associated optimal acceptance
rate in the interval $[0,0.234]$.

\section{Appendix}
\label{subsec:appendix}

In this Appendix, we prove the theorems stated in the paper.

\subsection{Proof of Theorem~\ref{Theorem:equivalence}}

Herein, assume the mixture distribution setting of \eqref{eq:themixer}
and \eqref{eq:themixerbeta} where the mixture components consist of
multivariate Gaussian distributions i.e.\ $g_j(x) =
\phi(x;\mu_j,\Sigma_j)$.
We prove each of the three parts of Theorem~\ref{Theorem:equivalence}
in turn.

\begin{proof}[Proof of Theorem~\ref{Theorem:equivalence}(a)]
Recall that $h_j(x)= w_j \phi(x;\mu_j,\Sigma_j)$ where $\exists C \in \mathbb{R}$ such that $C\sum_{j=1}^J w_j = 1$. Hence, taking $f_j(x,\beta) = [h_j(x)]^\beta$ gives 
\begin{eqnarray*} 
f_j(x,\beta) &=& w_j^\beta \phi(x;\mu_j,\Sigma_j)^\beta \\
&=& w_j^\beta \left[\frac{\left((2\pi)^d|\Sigma_j|\right)^{\frac{1-\beta}{2}}}{\beta^{d/2}} \right]\phi\left(x;\mu_j,\frac{\Sigma_j}{\beta}\right) \\
&\propto&  w_j^\beta |\Sigma_j|^{\frac{1-\beta}{2}}\phi\left(x;\mu_j,\frac{\Sigma_j}{\beta}\right)
\end{eqnarray*}
\end{proof}


\begin{proof}[Proof of Theorem~\ref{Theorem:equivalence}(b)]
Recall the result of Theorem~\ref{Theorem:equivalence}(a).
To adjust for the weight discrepancy from the cold state target a multiplicative adjustment factor, $\alpha_j(x)$ is used such that 
\[ f_j(x,\beta) = h_j(x)^\beta \alpha_j(x,\beta)\]
where $\alpha_j(x,\beta)= \left(w_j^\beta |\Sigma_j|^{\frac{1-\beta}{2}}\right)^{-1}$. An identical argument to
Theorem~\ref{Theorem:equivalence}(a)
shows that this immediately gives  $W_{(j,\beta)} \propto w_j$.

In a Gaussian setting, up to a proportionality constant
\begin{equation}
w_j \propto h_j(x)\left[(2\pi)^{\frac{d}{2}} |\Sigma_j|^{\frac{1}{2}} \exp\left\{ \frac{1}{2} (x-\mu_j)^T\Sigma_j^{-1}(x-\mu_j)   \right\}\right] \label{eqadjpap1}
\end{equation}
and at any point $x \in \mathbb{R}^d$ 
\begin{eqnarray}
\nabla\log{h_j(x)} &=& -\Sigma_j^{-1}(x-\mu_j) \label{eq:surr1} \\
\nabla^2 \log{h_j(x)}&=& -\Sigma_j^{-1}. \label{eq:surr2}
\end{eqnarray}
Substituting these gradient terms \eqref{eq:surr1} and \eqref{eq:surr2} into \eqref{eqadjpap1} and then using this form of \eqref{eqadjpap1} to create the adjustment factor $\alpha_j(x,\beta)= \left(w_j^\beta |\Sigma_j|^{\frac{1-\beta}{2}}\right)^{-1}$ completes the proof.
\end{proof}


\begin{proof}[Proof of Theorem~\ref{Theorem:equivalence}(c)]\\
Since $h_j(x)= w_j \phi(x;\mu_j,\Sigma_j)$ then
\begin{eqnarray}
f_j(x,\beta) &=& h_j(x)^\beta h_j(\mu_j)^{(1-\beta)} \nonumber\\
&=& w_j \phi(x;\mu_j,\Sigma_j)^\beta \phi(\mu_j;\mu_j,\Sigma_j)^{(1-\beta)} \nonumber \\
&=&  \frac{w_j}{\beta^{d/2}} \phi\left(x;\mu_j,\frac{\Sigma_j}{\beta}\right) \nonumber
\end{eqnarray}
and so  $W_{(j,\beta)} \propto w_j$.
\end{proof}

\begin{remark} It is possible to extend
the weight adjusted target result of Theorem~\ref{Theorem:equivalence}(c)
to a setting where the
target consists of a mixture of a general but common distribution, with
each component having a different shape and scale factor; we plan to
pursue this result elsewhere.
\end{remark}

Since mixing between modes is only possible at $\beta _{min}$, the dynamics will be identical to the single mode case ($J=1$) as covered in \cite{roberts2014minimising}.
It therefore follows directly  from Theorem~6 of
\cite{roberts2014minimising} that as $d\to\infty$,
the process $\{X_t\}$ converges weakly, at least on $X_t>0$,
to a diffusion limit $\{X_t\}_{t \ge 0}$ satisfying
\begin{eqnarray*}
dX_t  &=&  \left[2  \ell^2(X_t)
  \, \Phi\left( {- \ell(X_t) I^{1/2}(X_t) \over 2} \right) \right]^{1/2} dB_t
\big. \\
 && \hspace{-1cm}+ \, \Bigg[\ell(X_t) \, \ell'(X_t)
  \ \Phi \left({-I^{1/2}(X_t) \ell(X_t) \over 2}\right)
\big. \\
 && \hspace{-1cm}- \, \ell^2(X_t) \left({\ell(X_t) I^{1/2}(X_t) \over 2} \right)'
	\phi \left( {-I^{1/2}(X_t) \ell(X_t) \over 2} \right) \Bigg] dt
\, ,
\label{eq:origconv}
\end{eqnarray*}
where $I(\beta) =  \text{Var}_{x\sim f_1^\beta}(\log f_1(x))$.
If we extend the definition of $I$
to $I(\beta) =  \text{Var}_{x\sim f_1^\beta}(\log f_1(x))$ for $\beta>0$,
and $I(\beta) =  \text{Var}_{x\sim f_2^{|\beta|}}(\log f_2(x))$ for $\beta<0$,
so that positive values correspond to the first mode while negative
values correspond to the second mode,
then~\eqref{eq:origconv} also holds for $X_t<0$,
except with the sign of the drift reversed.

\subsection{Proof of Remark~\ref{Remark:Uniform}}

We note that for $x>0$ (with exactly analogous results for $x<0$),
$h'(x) = \ell(x)^{-1}$, and $h''(x) = -\ell'(x) \ell(x)^{-2}$.
So, if we set $H_t = h(X_t)$, then we compute by Ito's Formula that
\begin{eqnarray*}
dH_t
&=& h'(X_t) dX_t + \frac{1}{2}  h''(X_t) d\langle X \rangle_t \\
&=& \ell(X_t)^{-1} dX_t - \frac{1}{2} \ell'(X_t) \ell(X_t)^{-2}
d\langle X \rangle_t \\
&=& \ell(X_t)^{-1}
\left[2  \ell^2(X_t)
   \Phi\left( \frac{- \ell(X_t) I^{1/2}(X_t)}{ 2} \right) \right]^{1/2} dB_t\\
 && +  \ell(X_t)^{-1}  \Bigg[\ell(X_t)  \ell'(X_t)
   \Phi \left(\frac{-I^{1/2}(X_t) \ell(X_t)}{ 2}\right) \\
 && \hspace{-1cm}-  \ell^2(X_t) \left(\frac{\ell(X_t) I^{1/2}(X_t)}{ 2} \right)'
	\phi \left( \frac{-I^{1/2}(X_t) \ell(X_t)}{  2} \right) \Bigg] dt 
	\end{eqnarray*}
	\begin{eqnarray*}
- \frac{1}{2} \ell'(X_t) \ell(X_t)^{-2}
\left[2  \ell^2(X_t)
   \Phi\left( \frac{- \ell(X_t) I^{1/2}(X_t)}{ 2} \right) \right] dt
\end{eqnarray*}
	\begin{eqnarray*}
&=& 
\left[ 2 \Phi\left( \frac{- \ell(X_t) I^{1/2}(X_t)}{ 2} \right)
					\right]^{1/2} dB_t \\
  && +  \Bigg[\ell'(X_t)
  \ \Phi \left(\frac{-I^{1/2}(X_t) \ell(X_t)}{ 2}\right) \\
  && -  \ell(X_t) \left(\frac{\ell(X_t) I^{1/2}(X_t) }{ 2} \right)'
	\phi \left( \frac{-I^{1/2}(X_t) \ell(X_t)}{2} \right) \Bigg] dt 
		\end{eqnarray*}
		\begin{eqnarray*}
~~~~- \frac{1}{2} \ell'(X_t) 
\left[2 \,
  \Phi\left( {- \ell(X_t) I^{1/2}(X_t) \over 2} \right) \right] dt
\end{eqnarray*}
\begin{eqnarray}
&=&
\left[ 2 \, \Phi\left( {- \ell(X_t) I^{1/2}(X_t) \over 2} \right)
					\right]^{1/2} dB_t \nonumber\\
  && - \, \ell(X_t) \left({\ell(X_t) I^{1/2}(X_t) \over 2} \right)'
	\phi \left( {-I^{1/2}(X_t) \ell(X_t) \over 2} \right) dt \nonumber\\
&=&
\left[ 2 \, \Phi\left( {- \ell(X_t) I^{1/2}(X_t) \over 2} \right)
					\right]^{1/2} dB_t \nonumber\\
&& + \ \ell(X_t)
  \, \Bigg[ \Phi \left( {-I^{1/2}(X_t) \ell(X_t) \over 2} \right) \Bigg]' dt
	\label{dHt}
\end{eqnarray}

\noindent Re-writing everything in terms of $H_t = h(X_t)$, this becomes

$$
dH_t
\ = \
\left[ 2 \, \Phi\left( {- \ell(h^{-1}(H_t)) I^{1/2}(h^{-1}(H_t)) \over 2}
	\right) \right]^{1/2} dB_t
$$
\begin{equation}
  + \, \ell(h^{-1}(H_t))
  \, \Bigg[ \Phi \left( {-I^{1/2}(h^{-1}(H_t)) \ell(h^{-1}(H_t)) \over 2}
	\right) \Bigg]' dt
\, .
\label{eq:Hdef}
\end{equation}
Now, in general, a diffusion of the form $dH_t = \sigma(H_t) dB_t +
\mu(H_t) dt$ has locally invariant distribution $\pi$ provided that
$\frac{1}{2} (\log\pi)' \sigma^2 + \sigma \sigma' = \mu$.  In particular, it
has a {\it uniform} locally ivariant distribution , i.e.\ with $\pi$ constant,
provided that $\mu = \sigma \sigma'$, i.e.\ that $2 \mu = (\sigma^2)'$.
In this specific case, we verify that
$$
(\sigma^2)'
\ = \ {d \over dH} \left[ 2 \, \Phi\left( {- \ell(h^{-1}(H))
I^{1/2}(h^{-1}(H)) \over 2} 
        \right) \right]
$$
$$
\ = \ \left( {dH \over dX} \right)^{-1}
{d \over dX} \left[ 2 \, \Phi\left( {- \ell(X) I^{1/2}(X) \over 2} 
        \right) \right]
$$
$$
\ = \ \left( \ell(X)^{-1} \right)^{-1}
\left[ 2 \, \Phi\left( {- \ell(X) I^{1/2}(X) \over 2} 
        \right) \right]'
$$
which is indeed equal to $2\mu$ since in the above equation
$$
\mu \ = \ \ell(X)
  \, \Bigg[ \Phi \left( {-I^{1/2}(X) \ell(X) \over 2} \right) \Bigg]'
\, .
$$
Therefore $H$ leaves constant densities locally invariant.

\subsection{Proof of Theorem~\ref{Theorem:skewBM}}

We now assume that
$I(\beta) = I_0(\beta)/r_1$ for $\beta>0$, while
$I(\beta) = I_0(|\beta|)/r_2$ for $\beta<0$, and that
$\ell(\beta) = I_0^{-1/2}(\beta) \, \ell_0$.
This makes $~\ell(x) I^{1/2}(x) = \ell_0/\sqrt{r_1}~$ for $x>0$,
and $\ell(x) I^{1/2}(x)=\ell_0/\sqrt{r_2}$ for $x<0$.  In either
case, $\ell(x) I^{1/2}(x)$ is
constant, i.e.\ has derivative zero.  That in turn collapses~\eqref{dHt},
at least for $H_t \not= 0$, into the simpler
$$
dH_t
\ = \
\left[ 2 \, \Phi\left( {- \ell_0 \over 2 \sqrt{r(H_t)}}
				\right) \right]^{1/2} dB_t
\, ,
$$
where $r(H)=r_1$ for $H>0$ and $r(H)=r_2$ for $H<0$.

Finally, we set
$$
Z_t
\ = \ H_t \, \left[ 2 \, \Phi\left( {- \ell_0 \over 2 \sqrt{r(H_t)}}
                                \right) \right]^{-1/2}
.
$$
That is, $Z_t$ is a version of $H_t$ which is stretched by
a piecewise linear spatial function, which is linear on each of the
positive and negative values respectively.
It then follows immediately from the above
that $dZ_t = dB_t$, i.e.\ that $Z_t$
behaves like Brownian motion on each of its two branches (positive and
negative).
It remains only to prove that at $Z_t=0$, the convergence still holds.

We complete the proof similarly to previous proofs of diffusion limits of
MCMC algorithms (e.g.\ \cite{roberts1997weak, roberts1998optimal,
beda:rose:2008}), following the approach indicated in Chapter~4 of
\cite{ethi:kurt:1986} (in particular Corollary 8.7 of that chapter), by proving that the {\it generator} $G^{(d)}$ of the
original process under these combined transformations (i.e., jumping according to a rate $d$ Poisson process, then transformed by the $h$ function, and then stretched by the piecewise linear function) converges uniformly to the generator $G^*$ of skew Brownian motion, when applied to a core of functionals.

Let  $z_{\text{max}}=\int_{\beta_{\text{min}}}^1\ell^{-1}(u)\times
2\Phi(-\ell_0/(2\sqrt{r_1})$
and let $z_{\text{min}}=-\int_{\beta_{\text{min}}}^1\ell^{-1}(u)\times
2\Phi(-\ell_0/(2\sqrt{r_2})$.
We let $\mathcal{D}$ be the set of all functions $f:[-z_{\text{min}},z_{\text{max}}]\to \mathbb{R}$ which
are continuous and twice-continuously-differentiable on
$[z_{\text{min}},0]$ and also on $[0,z_{\text{max}}]$,
with matching one-sided second derivatives $f''^+(0)=f''^-(0)$,
and skewed one-sided first derivatives satisfying
$a f'^+(0)=bf'^-(0)$ where
$a = p
\left[ 2 \, \Phi\left( {- \ell_0 \over 2 \sqrt{r_1}} \right) \right]^{1/2}$
and
$b = (1-p)
\left[ 2 \, \Phi\left( {- \ell_0 \over 2 \sqrt{r_2}} \right) \right]^{1/2}$.
Finally we require that $f'(z_{\text{max}})=f'(z_{\text{min}})=0$ to describe the reflecting boundaries at the endpoints.
Thus, $C^2$ functions are not contained in $\mathcal{D}$ due the enforced
discontinuity of the first derivative at $0$,
but e.g.\ $f \in \mathcal{D}$ if $f(x) = x^2 + a x \mathbbm{1}_{x<0} + b x \mathbbm{1}_{x>0}$.
In particular, $\mathcal{D}$ is dense (in the sup norm) in the set of all
$C^2[z_{\text{min}},z_{\text{max}}]$ functions, so in the language of
\cite{ethi:kurt:1986}, $\mathcal{D}$ serves as a core of functions for which
it suffices to prove that the generators converge.
Furthermore, it follows from e.g.\ \cite{liggett:2010}
and Exercise~1.23 of Chapter~VII of \cite{revuz:yor:2004})
that the generator of skew Brownian motion (with excursion weights
proportional to $a$ and $b$ respectively) satisfies that
$G^*f(x) = \frac{1}{2} f''(0)$ for all $f\in \mathcal{D}$, using the convention that
$f''(0)$ represents the common value $f''^+(0) = f''^-(0)$.



Now, it follows from the previous discussion that for any
fixed $f\in \mathcal{D}$,
\begin{equation}
\lim_{d\to\infty}
\sup_{z \in \left[ z_{\text{min}},  z_{\text{max}} \right] \setminus \{0\}}
|G^{(d)}f(z) - G^* f(z)| = 0
\, .
\label{eq:genconv}
\end{equation}
That is, the generators do converge uniformly to the $G^*$, as required,
at least for $z \not= 0$, i.e.\ avoiding the mode-hopping value $\beta_{\text{min}}$.
To complete the proof, it suffices to prove that
\eqref{eq:genconv} also holds at $z=0$, i.e.\ to prove:

\begin{lemma}
We have that:
$$
\lim_{d\to\infty} G^{(d)}f(0) \ = \ G^* f(0)
\ = \ \, \frac{1}{2} \, f''(0)
\, .
$$
\end{lemma}

\begin{proof}
Note first that if the original inverse-temperature
process proposes to move the
inverse-temperature from $\beta_{\text{min}}$ to $\beta_{\text{min}} +
\ell(\beta_{\text{min}}) d^{-1/2}$,
then the $H_t$ process proposes to move from $0$ to $\pm d^{-1/2}$,
and the $Z_t$ process proposes to move from $0$ to $\pm d^{-1/2}
\left[ 2 \, \Phi\left( \frac{- \ell_0 }{ 2 \sqrt{r(\pm)}}
                                \right) \right]^{-1/2}$.
Furthermore, the $Z_t$ process, like the $X_t$ process, is sped up by
a factor of $d$, which multiplies its generator by $d$.  Hence, we
conclude that
$$
G^{(d)}f(0)
\ = \ d \left( p \alpha_+
\left[f \left(d^{-1/2} \left[ 2 \, \Phi\left( {- \ell_0 \over 2 \sqrt{r_1}}
                                \right) \right]^{-1/2} \right) - f(0) \right]
\right.
$$
$$
\, + \,
\left.
(1-p) \alpha_-
\left[f \left(d^{-1/2} \left[ 2 \, \Phi\left( {- \ell_0 \over 2 \sqrt{r_2}}
                                \right) \right]^{-1/2} \right) - f(0) \right]
\right)
\, ,
$$
where $\alpha_+$ is the acceptance probability for the original
process to accept a proposal
to increase the inverse-temperature from $\beta_{\text{min}}$ to
$\beta_{\text{min}}+\ell(\beta_{\text{min}}) d^{-1/2}$ in mode~1,
and $\alpha_-$ is the acceptance probability for
the same proposal in mode~2.



Next, note that the process $Z_t$ has expected squared jumping
distance equal to the square of its volatility, which is just equal to~1.
On the other hand, the expected squared jumping
distance must be equal to the squared
distance of its proposed move times the acceptance probability.
Hence, in mode~1, we must have
$$
1
\ = \ 
\left(
\left[ 2 \, \Phi\left( {- \ell_0 \over 2 \sqrt{r_1}}
                                \right) \right]^{-1/2}
\right)^2
\alpha_+
$$
whence
$$
\alpha_+ \ = \
2 \, \Phi\left( {- \ell_0 \over 2 \sqrt{r_1}} \right)
$$
and similarly
$$
\alpha_- \ = \
2 \, \Phi\left( {- \ell_0 \over 2 \sqrt{r_2}} \right)
\, .
$$

Then, taking a Taylor series expansion, we obtain that for $f \in \mathcal{D}$,
\begin{eqnarray*}
G^{(d)}f(0)
&=& d \Bigg( p \alpha_+
\left[f \left(d^{-1/2} \left[ 2 \Phi\left( {- \ell_0 \over 2 \sqrt{r_1}}
                                \right) \right]^{-1/2} \right) - f(0) \right]
 \\
&& \hspace{-2cm}+
(1-p) \alpha_-
\left[f \left( - d^{-1/2} \left[ 2 \, \Phi\left( {- \ell_0 \over 2 \sqrt{r_2}}
                                \right) \right]^{-1/2} \right) - f(0) \right]
\Bigg) 
\end{eqnarray*}
\begin{eqnarray*}
&=& d p
\left[ 2 \, \Phi\left( {- \ell_0 \over 2 \sqrt{r_1}} \right) \right]
\left(d^{-1/2} \left[ 2 \, \Phi\left( {- \ell_0 \over 2 \sqrt{r_1}}
                                \right) \right]^{-1/2} \right)
f'^+(0)\\
&& \hspace{-1cm}- \ \frac{1}{2} d p
\left[ 2 \, \Phi\left( {- \ell_0 \over 2 \sqrt{r_1}} \right) \right]
\left(d^{-1/2} \left[ 2 \, \Phi\left( {- \ell_0 \over 2 \sqrt{r_1}}
                                \right) \right]^{-1/2} \right)^2
f''^+(0) \\
&&+ O(d p \alpha_+ d^{-3/2})
\end{eqnarray*}
\begin{eqnarray*}
&+& \Bigg[d (1-p)
\left[ 2 \, \Phi\left( {- \ell_0 \over 2 \sqrt{r_2}} \right) \right]
\\
&& \times \left(d^{-1/2} \left[ 2 \, \Phi\left( {- \ell_0 \over 2 \sqrt{r_2}}
                                \right) \right]^{-1/2} \right)
f'^-(0) \Bigg]\\
&+&  \Bigg[\frac{1}{2} d (1-p)
\left[ 2 \, \Phi\left( {- \ell_0 \over 2 \sqrt{r_2}} \right) \right]
\\
&& \times\left(d^{-1/2} \left[ 2 \, \Phi\left( {- \ell_0 \over 2 \sqrt{r_2}}
                                \right) \right]^{-1/2} \right)^2
f''^-(0)\Bigg] \\
&+& O(d p \alpha_+ d^{-3/2})
\end{eqnarray*}
\begin{eqnarray*}
&=& d^{1/2} p
\left[ 2 \, \Phi\left( \frac{- \ell_0 }{ 2 \sqrt{r_1}} \right) \right]^{1/2}
f'^+(0) + \frac{1}{2} p f''^+(0)\\
&& - d^{1/2} (1-p)
\left[ 2 \, \Phi\left( \frac{- \ell_0}{ 2 \sqrt{r_2}} \right) \right]^{1/2}
f'^-(0)\\
&& +  \frac{1}{2} (1-p)
f''^-(0) + O(d^{-1/2}).
\end{eqnarray*}

Then, by definition of $f\in \mathcal{D}$, the terms involving $f'^+(0)$ and
$f'^-(0)$ cancel, and the terms
involving $f''^+(0)$ and $f''^-(0)$ combine.  Recalling the convention
$f''(0) = f''^+(0) = f''^-(0)$, we obtain finally that
$$
G^{(d)}f(0)
\ = \
\frac{1}{2}f''(0) + O(d^{-1/2})
$$
so that
$$
\lim_{d\to\infty}
G^{(d)}f(0)
\ = \
\frac{1}{2} f''(0)
\ = \
G^*f(0)
$$
as claimed.
\end{proof}

\bibliographystyle{spbasic}      
\bibliography{biblisim}   

%
%

\end{document}